\documentclass[a4paper]{article}
\usepackage{amsmath}
\usepackage{amssymb}
\usepackage{amsfonts}
\usepackage{amsthm}
\usepackage{amsbsy}
\usepackage{color}
\usepackage{graphicx}
\usepackage{enumerate}
\usepackage{subfigure}
\usepackage[utf8]{inputenc}
\usepackage{csquotes}
\usepackage[greek,english]{babel}
\usepackage[hyperref,doi,url=true,natbib,backend=biber,firstinits=true]{biblatex}

\AtEveryBibitem{
  \clearlist{address}
  \clearfield{date}
  \clearfield{eprint}
  \clearfield{isbn}
  \clearfield{issn}
  \clearlist{location}
  \clearfield{month}
  \clearfield{series}
  \clearlist{language}
  \clearfield{urldate}
  \ifentrytype{unpublished}{}{
    \clearfield{url}
    \clearfield{urldate}}
  \ifentrytype{book}{}{
    \clearlist{publisher}
    \clearname{editor}
  }
}

\DeclareFieldFormat{urldate}{}

\usepackage[bookmarks]{hyperref}

\newtheorem{thm}{Theorem}[section]
\newtheorem{cor}[thm]{Corollary}
\newtheorem{lem}[thm]{Lemma}
\newtheorem{prop}[thm]{Proposition}
\newtheorem{assumption}[thm]{Assumption}
\theoremstyle{definition}
\newtheorem{defn}[thm]{Definition}
\theoremstyle{remark}
\newtheorem{rmk}[thm]{\textbf{Remark}}

\numberwithin{equation}{section}

\newcommand{\E}{\mathbb{E}}
\renewcommand{\P}{\mathbb{P}}
\newcommand{\Q}{\mathbb{Q}}
\newcommand{\dif}{\,\mathrm{d}}
\newcommand{\df}{\mathrm{d}}
\newcommand{\wti}{\widetilde}

\newcommand{\poten}[1]{\mathrm{U}_{#1}}
\newcommand{\iden}{\boldsymbol{1}}
\newcommand{\half}{\dfrac{1}{2}}
\newcommand{\pian}[2]{\dfrac{\partial #1}{\partial #2}}

\newcommand{\R}{\mathbb{R}}

\newcommand{\as}{a.s.}

\newcommand{\ie}{i.e.}
\newcommand{\eg}{e.g.}
\newcommand{\dx}{\mathrm{d}x}
\newcommand{\dy}{\mathrm{d}y}
\newcommand{\ds}{\mathrm{d}s}
\newcommand{\dt}{\mathrm{d}t}
\newcommand{\dz}{\mathrm{d}z}
\newcommand{\pd}[2]{\frac{\partial #1}{\partial #2}}
\newcommand{\pdd}[3]{\frac{\partial^{#1} #2}{\partial #3^{#1}}}
\newcommand{\F}{\mathcal{F}}
\newcommand{\Fc}{\mathcal{F}}
\newcommand{\eps}{\varepsilon}
\newcommand{\indic}[1]{\boldsymbol{1}_{\{\ensuremath{#1}\}}}

\newcommand{\me}{\mathrm{e}}
\DeclareMathOperator{\supp}{supp}
\newcommand{\SEP}{{\bf{SEP}${}^*(\sigma,\nu,\mu)$}}
\newcommand{\OPT}{{\bf{OPT}${}^*(\sigma,\nu,\mu)$}}

\newcommand{\Ep}[1]{\E\left[#1\right]}

\renewcommand{\Pr}{\P}
\renewcommand{\vec}[1]{\mathbf{#1}}
\newcommand{\tds}{\bar{\tau}_D}

\title{Optimal robust bounds for variance options}

\author{Alexander M.~G.~Cox\thanks{Department of Mathematical
    Sciences, University of Bath, Bath, U.~K..
    e-mail: \texttt{a.m.g.cox@bath.ac.uk}, web:
    \texttt{http://www.maths.bath.ac.uk/$\sim$mapamgc/}}
  \and
  Jiajie Wang\thanks{Dipartimento di Metodi e Modelli per l’Economia, il Territorio e la Finanza,
  Universit\`a di Roma `La Sapienza', Italy. e-mail: \texttt{jiajie.wang@uniroma1.it}}}

\date{\today}

\begin{document}
\maketitle

\begin{abstract}
  Robust, or model-independent properties of the variance swap are
  well-known, and date back to \citet{Dupire:93} and
  \citet{Neuberger:94}, who showed that, given the price of
  co-terminal call options, the price of a variance swap was exactly
  specified under the assumption that the price process is
  continuous. In \citet{CoxWang:11} we showed that a lower bound on
  the price of a variance call could be established using a solution
  to the Skorokhod embedding problem due to \citet{Root:69}. In this
  paper, we provide a construction, and a proof of optimality of the
  upper bound, using results of \citet{Rost:71} and \citet{Chacon:85},
  and show how this proof can be used to determine a super-hedging
  strategy which is model-independent. In addition, we outline how the
  hedging strategy may be computed numerically. Using these methods,
  we also show that the Heston-Nandi model is `asymptotically extreme'
  in the sense that, for large maturities, the Heston-Nandi model gives
  prices for variance call options which are approximately the lowest
  values consistent with the same call price data.

\end{abstract}

\section{Introduction} \label{sec:Introduction}

The classical approach to derivative pricing problems is to hypothesise a
certain model for the underlying asset, and to invoke the Fundamental Theorem
of Asset Pricing to identify arbitrage-free prices with discounted expectations
under risk-neutral measures. In this setting, market information, for example
in the form of traded `vanilla' options, may be incorporated by choosing a
parametrised class of models, and determining the parameters by finding the
best fit to the observed prices.

An alternative approach to incorporating market information is to use the
traded options as part of a hedging strategy. This approach can be particularly
beneficial in the presence of model-risk, since, if carefully chosen, the
hedging properties of the strategy may still hold under a wide
class of models. The archetypal example of this is the hedging of a variance
swap using a $\log$-contract due to \citet{Dupire:93} and
\citet{Neuberger:94}. Suppose a (discounted) asset price has dynamics under the
risk-neutral measure:
\begin{equation} \label{eq:BasicModel}
\frac{\dif S_t}{S_t} = \sigma_t \, \dif W_t,
\end{equation}
where the process $\sigma_t$ is not necessarily known. A variance swap is a
contract where the payoff depends on the realised quadratic variation of the
log-price process: $\langle \ln S\rangle_T = \int_0^T \sigma^2_t \,
\dt$. Dupire and Neuberger observed that
\[
\dif(\ln S_t) = \sigma_t \, \dif W_t - \half \sigma_t^2 \, \dt
\]
from which we conclude that
\begin{equation*}
  \int_0^T \sigma^2_t \, \dt = 2\ln(S_0) - 2\ln(S_T) + 2 \int \frac{1}{S_t} \,
  \dif S_t.
\end{equation*}
It follows that one can replicate (up to a constant) the payoff of a
variance swap by shorting two log-contracts (that is, contracts which
pay $\ln(S_T)$ at maturity), and dynamically trading in the asset so
as to always hold $2/S_t$ units of the asset. The resulting portfolio
will hedge the variance swap under essentially {\it any} model of the
form \eqref{eq:BasicModel} (subject to very mild measurability and
integrability conditions on $\sigma_t$).

It follows that the price of a variance swap is essentially determined once one
observes the price of a log-contract. In practice, call options are more
liquidly traded, but the log contract can be statically replicated given a
continuum of call options, and so this information is more commonly used to
determine the price.

In this paper, we consider what can be said about the price of options
which pay the holder a function of the realised variance in the
presence of co-terminal call options. Two important examples of these
contracts are the variance call, which has payoff $\left(\langle \ln
  S\rangle_T -K\right)_+$ and the volatility swap, which has payoff
$\sqrt{\langle \ln S\rangle_T} - K$. Unlike in the case of the
variance swap, it is no longer possible to provide a trading strategy
which exactly replicates the payoff in any model, however, we are able
to provide both super- and sub-hedging strategies which work for a
large class of models, and therefore provide model-independent bounds
on the prices of such options. Our methods rely on techniques from the
theory of Skorokhod embeddings, and in particular, we need a novel
proof of optimality of some existing constructions. In one direction,
the bounds have been established in a preceding paper,
\citet{CoxWang:11}, and some of the methods used in this paper for the
other direction follow similar approaches, although there remain
substantial technical differences.

In addition to establishing theoretical bounds on the prices of these
options, we provide some numerical investigation of the bounds we
obtain, and a justification of our numerical techniques. Finally, we
are also able to establish, using the optimality techniques of
\cite{CoxWang:11}, a result relating to the extremality of the
Heston-Nandi model (the classical stochastic volatility model of
Heston, where the volatility and asset processes are perfectly
anti-correlated). Essentially, we show that for large maturities,
under the Heston-Nandi model, the price of a variance call option will
closely approximate the lowest price possible in the class of models
which produce identical call prices. Given that the Heston model is
commonly used for pricing options on variance, and that calibration
can often lead to values of the correlation parameter close to $-1$,
this suggests that using this model for pricing variance options
amounts to taking a strong `bet' on which model most accurately
reflects reality.

The theme of model-independent, or robust, pricing is one that has
received a great deal of attention in recent years. The approach we
take in this paper can be traced back to \citet{Hobson:98}, and more
recent work, closely related to variance options, includes
\citet{Dupire:05,CarrLee:10,DavisOblojRaval:10,CoxWang:11} and
\citet{OberhauserDosReis:13}. In addition, \citet{HobsonKlimmek:11}
consider variance swaps where the model may include jumps --- a case
which we exclude. An alternative, related approach is based on the
{\it uncertain volatility} models of \citet{Avellaneda:1995aa}. Recent
papers which take this approach include
\citet{Galichon:2011aa,possamai_robust_2013} and
\citet{Neufeld:2013aa}. We explain how our results may be interpreted
in this framework in Remark~\ref{rmk:uncertainvolatility}. Other
recent connected work in this direction includes
\citet{Beiglbock:2013aa} and \citet{Dolinsky:2013aa}, where
connections with optimal transport are established.

We proceed as follows: in Section~\ref{sec:financial-motivation} we
introduce our financial setup, and explain why the financial problem
of interest can be related to the Skorokhod embedding problem. This
motivates Section~\ref{sec:constr-rosts-barr}, where we provide a
characterisation of the solution of \citet{Rost:71} and
\citet{Chacon:85} to the Skorokhod embedding problem. In
Section~\ref{sec:optRost} we give a novel proof of the optimality of
these barriers, and explain how these constructions may be used to
derive superhedges for certain options. In
Section~\ref{sec:numerical-results} we show how the solutions may be
computed numerically, and provide some graphical evidence of the
behaviour of the hedging strategies. In
Section~\ref{sec:extr-hest-nandi} we prove our optimality result
for the Heston-Nandi model.

\section{Financial Motivation} \label{sec:financial-motivation}

To motivate our financial models, we begin with a fairly classical
setup: we suppose that there is a market which consists of a traded
asset, with price $S_t$ defined on a probability space $\left(\Omega,
  \Fc, \left( \Fc_t\right)_{t \ge 0}, \P\right)$, satisfying the usual
conditions, with:
\begin{equation} \label{eq:Sdefn} \frac{dS_t}{S_t} = r_t \dif t +
  \sigma_t \dif W_t
\end{equation}
under some probability measure $\Q \sim \P$, where $\P$ is the
objective probability measure, and $W_t$ a $\Q$-Brownian motion. In
addition, we suppose $r_t$ is the risk-free rate which we require to
be known, but which need not be constant. In particular, let $r_t,
\sigma_t$ be locally bounded, progressively measurable processes so that the
integral in \eqref{eq:Sdefn} is well defined, and so $S_t$ is an It\^o
process. We suppose that the process $\sigma_t$ is not known (or more
specifically, we aim to produce conclusions which hold for all
$\sigma_t$ in the class described). Specifically, we shall, at least
initially, suppose:
\begin{assumption} \label{ass:Price} The asset price process, under
  some probability measure $\Q \sim \P$, is the solution to
  \eqref{eq:Sdefn}, where $r_t$ and $\sigma_t$ are locally bounded,
  predictable processes.
\end{assumption}
We shall later see that some relaxation of this condition is possible,
using pathwise approaches to stochastic integration.

In addition, we need to make the following assumptions regarding the
set of call options which are initially traded:
\begin{assumption} \label{ass:Calls} We suppose that call options with
  maturity $T$, and at all strikes $\left\{K: K \ge 0\right\}$ are
  traded at time $0$, and the prices, $C(K)$, are assumed to be
  known. In addition, we suppose call-put parity holds, so that the
  price of a put option with strike $K$ is $P(K) = e^{-\int_0^T r_s
    \dif s}K -S_0 + C(K)$. We make the additional assumptions that
  $C(K)$ is a continuous, decreasing and convex function, with $C(0)=
  S_0$, $C_+'(0) = -e^{-\int_0^T r_s \dif s}$ and $C(K) \to 0$ as $K
  \to \infty$.
\end{assumption}
Many of these notions can be motivated by arbitrage concerns (see
\eg{} \citet{CoxObloj:11}). That there are plausible situations in
which these assumptions do not hold can be seen by considering models
with bubbles (\eg{} \cite{CoxHobson:05}), in which call-put parity
fails, and $C(K) \not\to 0$ as $K \to \infty$. Let us define $B_t =
e^{\int_0^t r_s \dif s}$, and make the assumptions above. Since
(classically) prices correspond to expectations under $\Q$, the
implied law of $B_T^{-1}S_T$ (which we will denote $\mu$) can be
recovered by the Breeden-Litzenberger formula
\citep{BreedenLitzenberger:78}:
\begin{equation} \label{eq:BL} \mu((K,\infty)) = \Q^*(B_T^{-1}S_T \in
  (K,\infty)) = - B_T C_+'(B_TK).
\end{equation}
Here we have used $\Q^*$ to emphasise the fact that this is only an
{\it implied} probability, and not necessarily the distribution under
the actual measure $\Q$.  It can now also be seen that the assumption
that $C_+'(0) = -B_T^{-1}$ is equivalent to assuming that there is no
atom at 0 --- \ie{} $\mu((0,\infty)) = 1$. In general, the equality
here could be replaced with an inequality ($C_+'(0) \ge -B_T^{-1}$) if
one wished to consider models with a positive probability of the asset
being worthless at time $T$. We do not impose the condition that the
law of $B_T^{-1} S_T$ under $\Q$ is $\mu$, we merely note that this is
the law implied by the traded options. We also do not assume anything
about the price paths of the call options: our only assumptions are
their initial prices, and that they return the usual payoff at
maturity. Finally, it follows from the assumptions that $\mu$ is an
integrable measure with mean $S_0$.

Our goal is to now to use the knowledge of the call prices to find a
lower or upper bound on the price of an option which has payoff
\begin{equation*}
  F\left(  \int_0^T \sigma_t^2 \, \dt\right) = F\left(
    \left< \ln S\right>_T \right).
\end{equation*}
The term $\int_0^T \sigma_t^2 \, \dt$ is commonly referred to as the
{\it realised variance}.  There are a number of pertinent examples
which motivate us: the most common case, and where the answer is well
known, is the case of a variance swap, where the payoff of the option
is $(\left< \ln S\right>_T - K)$. An obvious modification of this is
the {\it variance call} which has payoff $(\left< \ln S\right>_T -
K)_+$, and the corresponding variance put. In addition, {\it
  volatility swaps} are traded, where the payoff is $\sqrt{\left< \ln
    S\right>_T} - K$. As well as options written on the realised
variance, there are classes of options which trade on various forms of
weighted realised variance: define \begin{equation*} RV^w_T = \int_0^T
  w(S_t) \, \df \left< \ln S\right>_t = \int_0^T w(S_t) \sigma_t^2 \,
  \dt
\end{equation*}
then many of the above options can be recast in terms of their
weighted versions. Common examples of these include options on
corridor variance, where $w(x) = \indic{x \in [a,b]}$, and the gamma
swap \citep{Lee:2010aa}, where $w(x)= x$. In fact, for a
simplified exposition, we will assume that the weight depends not on
the spot price, but rather on the discounted spot price (or
equivalently, the forward price). In the case of most interest, where
$w(x) = 1$, this makes no difference, as it would, for example in
an equity setting where the dividend yield and the interest rate were
the same; we also refer to \citet{Lee:2010ac}, where it is indicated
how such an approximation may be accounted for using a
model-independent hedge involving calls of {\it all} maturities,
although this is beyond the general methodology described in the
article, where we will generally assume that only calls of one
maturity are observed.

Note we assume that our underlying price process, $S_t$, has
continuous paths. This is an important assumption, and our conclusions
will not generally hold otherwise. \citet{HobsonKlimmek:11} consider
related questions in the case where the underlying asset may jump. We
also assume that the payoff of the option is exactly the realised
quadratic variation, wheras in reality, financial contracts will be
written on a discretised version of the quadratic variation (for
example, the sum of squared daily log-returns); the effects of this
approximation are also considered by, for example,
\citet{HobsonKlimmek:11}.

Our approach is motivated by the following heuristics. Consider the
discounted stock price: \begin{equation*} X_t = e^{-\int_0^t r_s \dif
    s} S_t = B_t^{-1} S_t.
\end{equation*}
Under Assumption~\ref{ass:Price}, $X_t$ satisfies the SDE:
\begin{equation*}
  \df X_t = X_t \sigma_t \dif W_t.
\end{equation*}
Let $\lambda(x)$ be a strictly positive, continuous function, and
define a time change $\tau_t = \int_0^t \lambda(X_s) \sigma_s^2 \dif
s$. Writing $A_t$ for the right-continuous inverse, so that
$\tau_{A_t} = t$, we note that $\wti{W}_t = \int_0^{A_t} \sigma_s
\lambda(X_s)^{1/2} \dif W_s$ is a Brownian motion with respect to the
filtration $\wti{\F}_t = \Fc_{A_t}$, and if we set $\wti{X}_t =
X_{A_t}$, we have:
\begin{equation}\label{eq:Xdiffusiondef}
  \df \wti{X}_t = \wti{X}_t \lambda(\wti{X}_t)^{-1/2} \dif \wti{W}_t.
\end{equation}
In particular, under mild assumptions on $\lambda$, $\wti{X}_t$ is now
a diffusion on natural scale, and we note also that $\wti{X}_0 = S_0$
and $\wti{X}_{\tau_T} = X_T = B_T^{-1} S_T$. It follows that
$(\wti{X}_{\tau_T},\tau_T) = (B_T^{-1} S_T,\int_0^T \lambda(X_s)
\sigma_s^2 \, \ds)$, and therefore that \eqref{eq:BL}, which implies
knowledge of the law of $B_T^{-1} S_T$, also tells us the law of
$\wti{X}_{\tau_T}$. The key observation is that there is now a
correspondence between the possible joint laws of a stopped diffusion
and its stopping time, and the joint laws of the (discounted) asset
price at a fixed time, and the weighted realised variance at that
time. Since we wish to find the extremal possible prices of options
whose payoff is $F\left(\int_0^T w(X_t)\sigma^2_t \, \dt\right)$, if
we take $\lambda(x) \equiv w(x)$, the problem would appear to be
equivalent to that of finding a stopping time which maximises or
minimises $\E F(\tau)$ subject to $\mathcal{L}(\wti{X}_\tau) = \mu$,
where $\mu$ is the law of $B_T^{-1}S_T$ inferred by the market call
prices. The general problem of finding a stopping time for a process
which has a given distribution is known as the Skorokhod Embedding
problem, and solutions with given optimality properties have been well
studied in recent years
\citep{AzemaYor:79,Hobson:98,CoxHobsonObloj:08,CoxObloj:11b} --- for a
survey of these results, we refer the reader to
\citet{HobsonSurvey}. In \citet{CoxWang:11}, we established that a
construction of a Skorokhod embedding due to \citet{Root:69}
corresponded to minimising payoffs of the above form where $F(\cdot)$
is a convex increasing function. In this paper, we show that a related
construction, which can be traced back to work of \citet{Rost:71}, and
\citet{Chacon:85}, maximises such payoffs. In addition, we shall show
how these constructions may be used to derive trading strategies,
which will super- or sub-hedge in any of the models under
consideration, and are a hedge in the extremal model.

Finally, we will also briefly consider a similar problem where the
option is {\it forward starting} --- so the payoff depends on the
realised variance accumulated between future dates $T_0$ and $T_1$. In
this case, we assume that we observe call prices at times $T_0$ and
$T_1$, which correspond to distributions at times $T_0$ and $T_1$. The
implied distribution at time $T_0$, $\nu$ say, can be interpreted as
the law of $X_{T_0}$. Taking $\wti{X}_{t} = X_{T_0+\tau_t}$, we get a
problem similar to above, but with $\wti{X}_{0} \sim \nu$ for some new
probability measure $\nu$. This setting will be considered further in
Remark~\ref{rmk:ForwardStart}.

\section{Construction of Rost's Barrier}\label{sec:constr-rosts-barr}

\subsection{Background}

In this section, we recall some important results due to Rost and
Chacon concerning the construction of solutions to the Skorokhod
embedding problem. These results are established under fairly general
assumptions on the underlying process $\wti{X}_t$. One of the main
goals of this section is to provide conditions under which we can
apply their results in our setting.  We begin by recalling the notion
of a reversed barrier.
\begin{defn}[Reversed Barrier
  (\citep{Obloj:04})]\label{defn:rostbarrier}
  A closed subset $B$ of $[-\infty,+\infty]\times[0,+\infty]$ is a
  \emph{reversed barrier} if
  \begin{enumerate}[(i). ]
  \item\label{rostbarrier1} $(x,0)\in B$ for all
    $x\in[-\infty,+\infty]$;
  \item\label{rostbarrier2} $(\pm\infty,t)\in B$ for all
    $t\in[0,\infty]$;
  \item\label{rostbarrier3} if $(x,t)\in B$ then $(x,s)\in B$ whenever
    $s<t$.
  \end{enumerate}
\end{defn}
Given a reversed barrier, we can construct a stopping time of a
process $\wti{X}_t$ as $\tau = \inf\{t > 0: (\wti{X}_t,t) \in
B\}$. Then it is known that, given a measure $\mu$ satisfying certain
conditions, there exists a reversed barrier which embeds the law
$\mu$, that is, such that $\wti{X}_\tau \sim \mu$. Moreover, in the
case where $\wti{X}_t$ is a diffusion, the reversed barrier has the
property that the corresponding stopping time minimises the {\it
  capped expectation} $\E\left[ \tau \wedge t\right]$ over solutions
to the Skorokhod embedding problem, for all $t \ge 0$. Using the
observation that
\begin{equation*}
  F(\tau) = F(0) + \tau F'(0) + \int_0^\infty F''(t) \left( \tau -t
  \right)_+ \, \dt
\end{equation*}
and $(\tau-t)_+ = \tau - \tau \wedge t$, and the fact that $\E \sigma$
depends only on the law of $X_{\sigma}$ for `nice' stopping times (a
point we will elaborate on shortly), then it is immediate that $\E
F(\tau)$ is maximised over such stopping times for all convex
functions with $F'(0) = 0$.

These observations are essentially due to work of Rost and Chacon. In
\citet{Rost:71}, the notion of a {\it filling scheme} stopping time
was introduced for a general Markov process, and this was shown to
embed and have the optimality property described. Later,
\citet{Chacon:85}\footnote{We note that Chacon calls what we refer to
  a reversed barrier as a barrier. We follow the terminology
  established in \citet{Obloj:04}.} proved that in many cases, the
filling scheme stopping time would actually be almost-surely equal to
a stopping time generated by a reversed barrier. More recent results
concerning these constructions can be found in \cite{CoxPeskir:12}.

Before proving our results characterising the reversed barriers, we
recall some important background. Given a probability distribution
$\mu$ and a Markov process $\wti{X}$, the Skorokhod embedding problem is to
find a stopping time $\tau$ such that $\wti{X}_\tau\sim\mu$. Motivated
by the financial setting, we consider the case
that \begin{equation}\label{eq:positivemu} \mu\text{ is a probability
    distribution with } \mu((0,\infty))=1,
\end{equation}
and $\wti{X}$ is a regular (see \citet{RogersWilliams:00b} for
terminology relating to one-dimensional diffusions) diffusion on $I =
(0,\infty)$, which is a solution to \eqref{eq:Xdiffusiondef}, with
initial distribution $\wti{X}_0 \sim \nu$, for some given distribution
$\nu$, and a continuous function $\lambda$ on $I$ which is strictly
positive. Since $0 \not\in I$, $0$ is inaccessible for $\wti{X}_t$.

Recall that we wish to include the case of forward-starting options,
in which case $\wti{X}_0$ is assumed to have the law inferred from the
call prices at the start date of the contract. In the theory of
Skorokhod embeddings, it is usually natural to restrict to the class
of {\it minimal} stopping times, however since $\wti{X}_t$ is
transient, any embedding will be minimal. Moreover, for example by
considering $\wti{X}_t$ as a time change of a Brownian motion stopped
on hitting $0$, if we restrict ourselves to laws $\mu$ and $\nu$ which
have the same mean, then any embedding of $\mu$ must have $(\wti{X}_{t
  \wedge \tau})_{t \ge 0}$ a uniformly integrable (UI) process, and
moreover a necessary and sufficient condition for the existence of an
embedding is that
\begin{equation}\label{eq:potcond}
  \poten{\nu}(x) := - \int_{\R} |y-x|
  \, \nu(\df y) \ge - \int_{\R} |y-x|
  \, \mu(\df y) =: \poten{\mu}(x) > -\infty,
\end{equation}
for all $x\in\R$. By Jensen's inequality, such a constraint is clearly
necessary for the existence of an embedding; further, using
time-change arguments, and reducing to the Brownian case, it is the
only restriction that is required.  To understand this notion in the
financial setting, note that from \eqref{eq:BL} we deduce that
$\poten{\mu}(x) = S_0 - 2C(B_Tx) - x$, giving an affine mapping
between the function $\poten{\mu}(x)$ and the call prices.

For any given reversed barrier, we define $D:=(\R\times I)\backslash
B$, and then one can show that there exists a unique upper
semi-continuous function $R:I\rightarrow[0,\infty]$ such that
\begin{equation*}
  D = \{(x,t): t> R(x)\}\ \ \text{ and }\ \ B = \{(x,t): 0 \leq t\leq R(x)\}.
\end{equation*}
Then the stopping time of interest is:
\begin{equation*}
  \tau_D = \inf\{t>0:\,(\wti{X}_t,t)\notin D\} = \inf\{t>0: t \leq R(\wti{X}_t)\},
\end{equation*}
and we will call such an embedding a {\it Rost stopping time}, or a
{\it Rost embedding}. We may also refer to $D$ as a reversed barrier,
with the meaning intended to be inferred from the notation.
Note that multiple barriers may solve the same stopping problem: for
example, if the target distribution contains atoms, the barrier
between atoms may be unspecified away from the starting point,
provided it is never beyond the `spikes' of the atoms.

We can put together the work of Rost and Chacon in the specific case
that the underlying process is a regular diffusion on $I$, satisfying
a certain regularity condition needed to ensure Chacon's result holds.
Specifically, we introduce the set:
\begin{equation}
  \label{eq:LambdaCond}
  \mathcal{D} = \left\{ \lambda \in C(I;\R)
    : \parbox{8cm}{$\lambda$ is strictly positive and
      the solution to \eqref{eq:Xdiffusiondef} defines a 
      regular diffusion on $I$, with transition density $p(t,x,y)$
      with respect to Lebesgue such that, for any $x_0 \in I$,
      $c>0$,  open set
      $A$ containing $x_0$ and $\eps >0$, there exists $\delta
      >0$ such that $|\left(p(x,x_0,t)-p(x,x_0,s)\right)x_0^2
      \lambda(x_0)^{-1}|<\eps$ whenever $|s-t| < 
      \delta$ and either $x_0 \not\in A$ or $t>c$.}\right\}
\end{equation}
Then we can prove the following result:
\begin{thm} \label{thm:RostChacon} Suppose $\mu$ and $\nu$ are
  probability measures on $I$ such that \eqref{eq:potcond} holds, and
  suppose $\wti{X}_t$ solves \eqref{eq:Xdiffusiondef} for some
  $\lambda \in \mathcal{D}$ with $\wti{X}_0 \sim \nu$:
  \begin{enumerate}
  \item if $\mu$ and $\nu$ have disjoint support, then there exists a
    reversed barrier $D$ such that $\wti{X}_{\tau_D} \sim \mu$;
  \item if $\mu$ and $\nu$ do not have disjoint support, then (on a
    possibly enlarged probability space) there exists a random
    variable $S \in \{0,\infty\}$, and reversed barrier $D$, such that
    $\wti{X}_{\tau_D \wedge S} \sim \mu$.
  \end{enumerate}
  Moreover, in both cases, the resulting embedding maximises $\E
  F(\sigma)$ over all stopping times $\sigma$ with $\wti{X}_\sigma
  \sim \mu$ and $\E \sigma = \E \tau_D < \infty$, for any convex
  function $F$ on $[0,\infty)$.
\end{thm}

That the condition $\E \sigma = \E \tau_D$ is reasonable can be seen
by considering
\begin{equation}\label{eq:Qvers1}
  Q(\wti{X}_t) := \int_{x_0}^{\wti{X}_t} \int_{x_0}^y \lambda(z) z^{-2} \, \dz \dy =
  \int_0^t \int_{x_0}^{x_s} \lambda(y) y^{-2} \, \dy \df \wti{X}_s + \half t + Q(\wti{X}_0).
\end{equation}
Noting that $Q(x)$ is convex, we can take expectations along a
localising sequence and apply Fatou's Lemma to see that, for any
stopping time $\tau$, $\E \tau \ge \E Q(\wti{X}_\tau)$, the second
term depending only on the law of $\wti{X}_\tau$. For well behaved
stopping times (specifically, where $\E Q(\wti{X}_{\tau\wedge N}) \to
\E Q(\wti{X}_{\tau})$), we would in fact expect equality here.

We also observe that there is a trivial extension of this result when
$F(\cdot)$ is a concave function, then $-F(\cdot)$ is a convex
function, and so the resulting stopping time minimises $\E F(\sigma)$
over the same class of stopping times.

\begin{proof}[Proof of Theorem~\ref{thm:RostChacon}]
  We first show that, under the conditions above, there exists a
  filling scheme stopping time.

  Standard time-change arguments, and reduction to the Brownian case
  show that when $\mu$ and $\nu$ have the same mean, then
  \eqref{eq:potcond} is both necessary and sufficient for the
  existence of a Skorokhod embedding which solves the problem. In
  addition, by
  \citet[Theorem~4]{Rost:71}, this is sufficient to deduce the
  existence of a filling scheme stopping time: 
  from the proof of this result, it is clear that whenever an
  embedding exists, it can be taken as a filling scheme stopping time.

  Now consider the case where $\mu$ and $\nu$ have disjoint
  support. Then \citet[Theorem~3.24]{Chacon:85} states that a filling
  scheme stopping time is a reversed barrier stopping time provided:
  \begin{enumerate}
  \item $\wti{X}_t$ is a standard Markov process, in duality with a
    standard Markov process $\hat{X}_t$ (we refer the reader to \eg{}
    \citet[Definition~VI.1.2]{BlumenthalGetoor:68});
  \item the transition measures relating to $\wti{X}_t$ and
    $\hat{X}_t$ have densities;
  \item\label{item:1} the transition density $p(x,y,t)$ for
    $\wti{X}_t$ satisfies an equicontinuity property: for any $x_0 \in
    \R$, $c>0$ and open set $A$ containing $x_0$, then given $\eps
    >0$, there exists $\delta >0$ such that
    $|p(x,x_0,t)-p(x,x_0,s)|<\eps$ whenever $|s-t| < \delta$ and
    either $x_0 \not\in A$ or $t>c$.
  \end{enumerate}
  Since $\wti{X}_t$ is a regular diffusion with inaccessible
  endpoints, it is its own dual process, with respect to the speed
  measure (which is $\lambda(x) x^{-2} \dx$), see
  \citet[Remark~1.15]{Fitzsimmons:98}. Moreover, under these
  conditions, the transition density exists (\eg{}
  \citet[Theorem~V.50.11]{RogersWilliams:00b}), and
  \eqref{eq:LambdaCond} guarantees that \ref{item:1} holds.

  Finally, suppose $\mu$ and $\nu$ are not disjointly supported. Then,
  using the Hahn-Jordan decomposition, we can find disjoint,
  non-negative measures $\nu_0$ and $\mu_0$, such that $\nu_0(A) \le
  \nu(A)$, $\mu_0(A) \le \mu(A)$, for all $A \in
  \mathcal{B}(0,\infty)$, and $\mu-\nu = \mu_0-\nu_0$. Since $\mu$ and
  $\nu$ are not disjoint, $\mu_0$ and $\nu_0$ are non-trivial. We
  write in addition $\nu \wedge \mu = \nu - \nu_0$, observing that
  also then $\nu \wedge \mu = \mu - \mu_0$. By enlarging the probability
  space if necessary, let $Z$ be a uniform random variable on $(0,1)$,
  independent of the process, and define the Radon-Nikodym derivative
  $f = \dif (\nu \wedge \mu)/ \dif \nu$. Then if we set
  \begin{equation*}
    S =\begin{cases}
      0,\ &\text{if }\ Z\leq f(\wti{X}_0);\\
      \infty,\ &\text{if }\ Z> f(\wti{X}_0),
    \end{cases}
  \end{equation*}
  it follows that $\P(X \in A, S = 0) = (\nu \wedge \mu) (A)$. Now
  define the normalised measures, $\nu^*(A) = \nu_0(A) /
  \nu_0((0,\infty))$ and $\mu^*(A) = \mu_0(A) / \mu_0((0,\infty))$,
  and construct the reversed barriers for the initial distribution
  $\nu^*$ and target distribution $\mu^*$. (It is straightforward to
  check that $U_{\nu_0}(x) - U_{\mu_0}(x) = U_{\nu}(x) - U_{\mu}(x)$,
  and therefore that the construction is possible.) However, it is now
  clear that this embeds and is exactly the stopping time described in
  the statement of the theorem. Moreover, this description of the
  first step of the construction is exactly the first step described
  in the construction of the filling scheme by Rost, so it follows
  that this stopping time is a filling scheme stopping time.

  The final statement is \citet[Proposition~2.2]{Chacon:85}.
\end{proof}

It is clear that the above conditions include the main case of
interest --- the case where $\lambda(x) = 1$, which is the case
corresponding to options on realised variance. For the point $0$ to be
inaccessible, we require $\int_{0+} \lambda(x) x^{-1} \, \dx = \infty$
\citep[Theorem~V.51.2]{RogersWilliams:00b}. In principle, this would
exclude, for example, the case where $\lambda(x) = x$, the Gamma swap,
however in practice, this case could be approximated by taking
$\lambda(x) = x \vee \eps$ for some small $\eps>0$ . Note however that
in the case where $\lambda(x) = x$, $\wti{X}_t$ is a Bessel process of
dimension $0$, and so in particular, the process will hit zero with
positive probability; this is not crucial, since we assume our target
measure has no atom of mass at $0$, and so we would expect the
reversed barrier to stop the process before hitting zero, and
therefore the exact behaviour near zero should not affect the barrier substantially.

\subsection{Construction of reversed barriers}

In this section, we show how the reversed barrier determined in
Theorem~\ref{thm:RostChacon} can be constructed.  As above, we suppose
that we have a time-homogeneous diffusion, $\wti{X}_t$, such that:
\begin{equation}
  \label{eq:Xdefn}
  \begin{split}
    \df \wti{X}_t & = \sigma(\wti{X}_t) \df W_t,\\
    \wti{X}_0 & \sim \nu,
  \end{split}
\end{equation}
so $\sigma(x) = x \lambda(x)^{-1/2}$. In addition, we suppose that the
diffusion coefficient, $\sigma: I \to (0,\infty)$ is a continuously
differentiable function such that:
\begin{equation}
  \label{eq:SigmaCond}
  x^2 \sigma(x)^{-2} \in \mathcal{D}, \quad |\sigma(x) x^{-1}| \text{ and } |\sigma'(x) \sigma(x) x^{-1}| \text{
    are bounded on } (0,\infty),
\end{equation}
or equivalently, that $\lambda: I \to (0,\infty)$ is continuously
differentiable, and
\begin{equation}
  \label{eq:LambdaCond2}
  \lambda(x) \in \mathcal{D}, \quad |\lambda(x)^{-1}| \text{ and } |\lambda'(x) \lambda(x)^{-2} x| \text{
    are bounded on } (0,\infty).
\end{equation}
Then our general problem is:
\begin{center}
  \begin{tabular}{rp{9cm}}
    \SEP{}: & Find an upper-semicontinuous
    function $R(x)$ such that
    the domain $D= \{(x,t): t > R(x)\}$ has $\wti{X}_{\tau_D\wedge S} \sim
    \mu$, where $\wti{X}_t$ is given by \eqref{eq:Xdefn},
    $\tau_D = \inf\{t>0: (\wti{X}_t,t) \not\in D\} = \inf\{t>0: t \le
    R(\wti{X}_t)\}$, and $S \in \{0,\infty\}$ is an $\Fc_0$-measurable random
    variable such that $\wti{X}_0 \sim \nu \wedge \mu$ on $\{S=0\}$.
  \end{tabular}
\end{center}
Here, the measure $\nu \wedge \mu$ is as defined in the proof of
Theorem~\ref{thm:RostChacon}. We restrict ourselves to the case where
$\poten{\nu}(x) \ge \poten{\mu}(x)$. We will also introduce the
notation $\tds = \tau_D \wedge S$.

Then we have the following result:
\begin{thm}
  Suppose \eqref{eq:Xdefn} and \eqref{eq:SigmaCond} hold. Assume $D$
  solves \SEP{}.  Then 
  $u(x,t) = \poten{\mu}(x) + \E^\nu \left| x-\wti{X}_{t\wedge\tds}\right|$ is
  the unique bounded viscosity solution to:
  \begin{subequations}\label{eq:ViscSoln}
    \begin{align}
      \label{eq:ViscSolnMain}
      \pd{u}{t}(x,t) & = \left( \frac{\sigma(x)^2}{2} \pd{^2
          u}{x^2}(x,t)\right)_+ \\
      u(0,x) & = \poten{\mu}(x) -
      \poten{\nu}(x). \label{eq:ViscSolnBC}
    \end{align}
  \end{subequations}
  Moreover, given the solution $u$ to \eqref{eq:ViscSoln}, a reversed
  barrier $D$ which solves \SEP{} can be recovered by $D = \{ (x,t) :
  u(x,t) > u(0,t)\}$.
\end{thm}

\begin{rmk}
  In a recent paper, \citet{OberhauserDosReis:13} make a very similar
  observation, and they also use a viscosity solution approach to
  derive existence and uniqueness in the Rost setting (unlike
  \cite{CoxWang:11} where a variational inequality-based approach is
  taken). They also work in the slightly more general setting where
  the diffusion coefficient $\sigma$ may depend both on time and
  space; it seems very likely that the results should extend to this
  setting, but we observe that the optimality of such a construction
  is no longer easily determined; given that we are interested in the
  optimality properties of such processes, we restrict ourselves to
  the time-homogenous case.
\end{rmk}

We wish to use standard results on viscosity solutions from
\citet{FlemingSoner:06}. The equation \eqref{eq:ViscSoln} is really a
forward equation, rather than a backward equation, which is the
setting in \cite{FlemingSoner:06}; however we can apply their results
to the function $v(x,T-t) = -u(x,t)$, for a fixed $T>0$. Since $T$ is
arbitrary, the extension to an infinite horizon is straightforward.

\begin{proof}
  We first note that since $\mu$ and $\nu$ are integrable, the
  function $u$ is bounded, both $-\poten{\mu}(x)$ and $h(x,t) =
  \E^{\nu}|x-\wti{X}_{t \wedge \tds}|$ are convex, continuous
  functions, and so their second derivatives (in $x$) exist as
  positive measures. Moreover, $-\poten{\mu}''(x) = 2 \mu(\dx)$, and
  we can also decompose $h_{xx}(x,t)/2$ into two measures $\mu_t^1$ and
  $\mu_t^2$ defined by:
  \begin{align*}
    \int f(x) \, \mu_t^1 (\dx) & = \E^\nu \left[f(\wti{X}_t); t <
      \tds\right]\\
    \int f(x) \, \mu_t^2 (\dx) & = \E^\nu \left[f(\wti{X}_{\tds}); t
      \ge \tds\right].
  \end{align*}
  In particular, since $\tds$ embeds $\mu$, we must have $\mu_t^2(A)
  \le \mu(A)$ for all $A$, and $\mu_t^2(A) = \mu(A)$ for all $A
  \subseteq \{x | t \ge R(x)\}$. In addition, $\mu_t^1$ will be
  dominated by the transition density of a diffusion (which exists by
  \eqref{eq:LambdaCond}) started with distribution $\nu$, and so will
  have a density $f^1(x,t)$ with respect to Lebesgue for all $t>0$,
  and since $D$ is open, by a slight modification of
  \citet[Lemma~3.3]{CoxWang:11}, it is easily checked that $u(x,t)$
  satisfies $\pd{u}{t}(x,t) =\frac{\sigma(x)^2}{2} \pd{^2
    u}{x^2}(x,t)$ in $D$. We observe also that $f^1(x_n,t_n) \to 0$
  whenever $(x_n,t_n) \to (x,t) \in B$, since $f^1$ is dominated by
  the density of a diffusion with initial law $\nu$, killed if it hits
  $x$ before $t$, which also has this property.
  
  The result now follows from
  \citet[Proposition~V.4.1]{FlemingSoner:06}, when we observe that for
  $t\le R(x)$, any smooth function $w$ such that $w-u$ is a local
  minimum at $(x,t)$ must have $\pd{w}{t} = 0$. For $t < R(x)$ this
  follows from observing that $\E^\nu \left|
    x-\wti{X}_{t\wedge\tds}\right|$ is constant in $t$ whenever
  $t<R(x)$. For $t=R(x)$, we observe that $f^1(x,t)$ can be made
  arbitrarily small in a ball near $(x,t)$, and $\pd{^2u}{x^2} =
  f^1(x,t)$ in $D$. Hence, for such $w$, we have $-w_t +
  \left(\sigma^2(x) w_{xx}\right)_+ \ge 0$. For smooth $w$ such that
  $w-u$ has a local maximum at $(x,t)$, we first observe that for $t
  \le R(x)$, the argument above implies that $u_{xx}(x,t) \le 0$, and
  so $w_{xx}(x,t) \le 0$ also. In addition, by Jensen's inequality,
  $u(x,t)$ is non-decreasing, so $w_{t}(x,t) \ge 0$. It follows that
  $-w_t + \left(\sigma^2(x) w_{xx}\right)_+ \le 0$.

  It follows that $u$ is indeed a bounded viscosity solution. To see
  that it is unique, we apply \citet[Theorem~V.9.1]{FlemingSoner:06}
  to $u(t,\me^y)$ (noting also the comment immediately preceding this
  proof). It is now routine to check that \eqref{eq:SigmaCond} is
  sufficient to ensure that there is a unique bounded viscosity
  solution to \eqref{eq:ViscSoln}.

  To conclude that a reversed barrier can be recovered from a solution $u$ to
  \eqref{eq:ViscSoln}, we observe that Theorem~\ref{thm:RostChacon} guarantees
  the existence of a reversed barrier $D$, and by Jensen's inequality, we see
  that $D^* = \{(x,t): u(x,t) > u(t,0)\}$ does indeed define a reversed
  barrier. From the arguments above, we also conclude that $\Pr(\tds
  = \bar{\tau}_{D^*}) = 1$, since $\wti{X}_{\tds}$ is supported on the set
  where $u(x,t) = u(x,0)$, and $f(x,t) = 0$ on this set also.
\end{proof}

\begin{rmk}
  We observe also that the solution $u(x,t)$ has an interpretation in terms of
  an optimal stopping problem \citep[c.f.][Remark~4.4]{CoxWang:11}. Fix $T>0$ and
  set
  \begin{equation*}
    v(x,t) = \sup_{\tau \in [t,T]} \E^{(x,t)}\left[\poten{\mu}(\wti{X}_\tau) -
      \poten{\nu}(\wti{X}_{\tau})\right],
  \end{equation*}
  where the supremum is taken over stopping times $\tau$, and the expectation
  is taken conditional on $\wti{X}_t = x$. Then standard results for optimal
  stopping problems suggest that $v(x,t)$ is the solution to the viscosity
  equation:
  \begin{equation}\label{eq:OptStopVisc}
    \begin{split}
      & \max\left\{\half \sigma(x)^2 \pd{^2 v}{x^2}(x,t) + \pd{v}{t}(x,t),
        v(x,t)-(\poten{\mu}(x) - \poten{\nu}(x))\right\} = 0,\\ & v(x,T) =
      \poten{\mu}(x) - \poten{\nu}(x).
    \end{split}
  \end{equation}
  Now observe from the problem formulation that $v(x,t)$ is certainly
  decreasing in $t$ (a stopping time which is feasible for $t_1$ is
  also feasible for $t_0<t_1$, with the same reward). Using the fact
  that both the solution to \eqref{eq:ViscSoln} and
  \eqref{eq:OptStopVisc} are monotone in $t$, it is possible to deduce
  that $v(x,T-t)$ solves \eqref{eq:ViscSoln}, and $u(x,T-t)$ solves
  \eqref{eq:OptStopVisc}, so that they must be the same function.
\end{rmk}

\section{Optimality of Rost's Barrier, and superhedging strategies}
\label{sec:optRost}

\subsection{Optimality via pathwise inequalities}
\label{sec:optim-via-pathw}

For a given distribution $\mu$, Theorem~\ref{thm:RostChacon} says that
Rost's solution is the ``maximal variance'' embedding.  A slight
generalisation of this result leads us to consider the following
problem:
\begin{center}
  \begin{tabular}{rp{8.9cm}}
    \OPT{}:
    & Suppose \eqref{eq:positivemu} and \eqref{eq:potcond} hold,
    and $\wti{X}$ solves \eqref{eq:Xdefn}. Find a stopping time
    $\tau$ such that $\wti{X}_\tau \sim \mu$ and, for
    a given increasing convex function $F$ with $F(0)=0$,
    \begin{equation*}
      \E\big[ F(\tau)\big]
      \ =\ \sup_{\rho: \wti{X}_{\rho} \sim \mu}
      \E\big[ F(\rho)\big].
    \end{equation*}
  \end{tabular}
\end{center}

Our aim in this section is to find the super-replicating hedging
strategy for call-type payoffs on variance options, however it is not
immediately obvious how to recover such an identity directly from the
proofs of the optimality criterion given in \citet{Chacon:85}. Rather,
we shall provide a `pathwise inequality' which encodes the optimality
in the sense that we can find a supermartingale $G_t$, and a function
$H(x)$ such that
\begin{equation}
  \label{eq:pathwiseineq} F(t) \le G_t + H(\wti{X}_t)
\end{equation}
and such that, for Rost's embedding $\tds$, equality holds in
\eqref{eq:pathwiseineq} and $G_{t \wedge \tds}$ is a UI martingale.
It then follows that $\tds$ does indeed maximise $\E F(\tau)$ among
all solutions to the Skorokhod embedding problem, and further, using
\eqref{eq:pathwiseineq}, we can super-replicate call-type payoffs on
variance options by a dynamic trading strategy.

Suppose that we have already found the solution to \SEP,
$\tds$. Define the function
\begin{equation}
  \label{eq:Mdefn}
  M(x,t) = \E^{(x,t)}\left[f(\tds)\right],
\end{equation}
where $f$ is the left derivative of $F$ and $\tds$ is the corresponding hitting
time of $B$. Specifically, observe that $\wti{X}_t$ is a Markov process, and we
interpret the expectation as the average of $f(\tds)$ given that we start at
$\wti{X}_t = x$, with $\tds \ge t$. At zero, given the possibility of stopping
at time $0$, it is not immediately clear how to interpret the conditioning ---
it will turn out not to matter, but a natural choice would be to replace $\tds$
with $\tau_D$. In the following, we shall assume:
\begin{equation}\label{eq:Massumption}
  M(x, t)\text{ is locally bounded on }\R \times \R_+.
\end{equation}
Obviously, $M(x,t)=f(t)$ whenever $0\leq t<R(x)$.  Now given a fixed time
$T>0$, and choosing $S_0^* \in (\inf \supp(\mu),\sup \supp(\mu))$ (which is
non-empty if $\mu$ is non-trivial), we define
\begin{equation*}
  Z_{T}(x)\ =\ 2\int_{S_0^*}^{x}\int_{S_0^*}^{y}
  \dfrac{M(z,T)}{\sigma(z)^2}\dif z\dif y,
\end{equation*}
and in particular, $Z_{T}''(x)=2M(x,T)/\sigma^2(x)$, and $Z_T$ is a
convex function. Define also
\begin{align}
  G_{T}(x,t) &= F(T)-\int_{t}^{T}M(x,s)\dif s-Z_{T}(x) \nonumber\\
  H_{T}(x) &= \int_{R(x)}^{T}\Big[M(x,s)-f(s)\Big]\dif s+Z_{T}(x) \nonumber\\
  &= \int_{R(x)\wedge T}^{T}\Big[M(x,s)-f(s)\Big]\dif s+Z_{T}(x).\nonumber\\
  Q(x) & = \int_{S_0^*}^x\int_{S_0^*}^y\dfrac{2}{\sigma(z)^2}\dif z\dif
  y. \label{eq:Qdefn}
\end{align}
Then we have the following results
\begin{prop}
  \label{prop:GTineq}
  For all $(x,t,T)\in\R_+\times\R_+\times\R_+$, we have,
  \begin{equation}
    \label{eq:GTineq}
    \begin{cases}
      \ G_{T}(x,t)\,+\,H_{T}(x)\ \geq\ F(t),\ \ \ &\text{if }\ t>R(x)\,;\\
      \ G_{T}(x,t)\,+\,H_{T}(x)\ =\ F(t),&\text{if }\ t\leq R(x)\,.
    \end{cases}
  \end{equation}
\end{prop}

\begin{lem}
  \label{lem:GTmartsubmart}
  Under the setting \eqref{eq:positivemu} -- \eqref{eq:potcond}, suppose that
  the stopping time $\tds$ is the solution to \SEP{}. Moreover, assume $f$ is
  bounded and
  \begin{equation}
    \label{eq:qvTbddcond}
    \text{for any }\ T>0,\ \ \left(Q(\wti{X}_t);\,0\le t\le T\right)
    \ \text{ is a uniformly integrable family. }
  \end{equation} 
  Then for any $T>0$, the process
  \begin{equation} \label{eq:GTmart}
    \left(G_T(\wti{X}_{t\wedge\tds},t\wedge\tds);\,0\le t\le
      T\right) \ \text{ is a martingale,}
  \end{equation}
  and
  \begin{equation}\label{eq:GTsupmart}
    \left(G_T(\wti{X}_{t},t);\,0\le t\le T\right)
    \ \text{ is a supermartingale.}
  \end{equation}
\end{lem}

We note that when $\sigma(x) = x$, so $\wti{X}_t$ is geometric Brownian motion,
then it is straightforward to check that, for all $T>0$, $\sup_{t\le
  T}\E[Q(\wti{X}_t)^2]<\infty$, and so \eqref{eq:qvTbddcond} is trivially
satisfied provided $\E[Q(\wti{X}_0)^2] < \infty$. Note also that since
$\wti{X}_t$ is a local martingale bounded below, for any embedding $\tau$ which
embeds $\mu$ we have $\E[\wti{X}_\tau]= \E[\wti{X}_{\tds}]$. It follows that if
$\E[\wti{X}_{\tds}] =\E[\wti{X}_0]$, any embedding of $\mu$ is a martingale,
and not just a local-martingale.

Then the main result of this section follows.

\begin{thm}
  \label{thm:optRost} Suppose that $\tds$ is
  the solution to \SEP{}, and \eqref{eq:qvTbddcond} holds, then
  $\tds$ solves \OPT{}. \end{thm}

\begin{proof} We first consider the case where $\E[\tds]=\infty$.
  Since $F(t)\geq\alpha+\beta t$ for some constants $\alpha\in\R$ and
  $\beta\in\R_+$, we must have $\E[F(\tds)]=\infty$. The result is
  trivial. So we always assume $\E[\tds]<\infty$.

  Under the assumption $\E[\tds]<\infty$, consider $Q(\cdot)$ given
  by \eqref{eq:Qdefn}. We have (recall \eqref{eq:Qvers1})
  $\E[Q(\wti{X}_{\tds})]=\E[\tds] + \E[Q(\wti{X}_0)]<\infty$. Therefore, for all
  $\tau$ embedding $\mu$, $\E[\tau] = \E[Q(\wti{X}_\tau)]
  =\E[Q(\wti{X}_{\tds})]=\E[\tds]+ \E[Q(\wti{X}_0)]<\infty $. In the remainder of
  this proof, we always assume $\E[\tau]=\E[\tds]<\infty$.

  We first assume $f$ is bounded, since $f$ is increasing,
  \begin{equation}\label{eq:fassumption}
    \text{there exists }C<\infty,\text{ such that }\lim_{t \to \infty}
    f(t) = C.
  \end{equation}
  For $T>0$, since $M(\cdot,T)$ is also bounded by $C$, then
  \[\E[Z_T(\wti{X}_{t\wedge\tau})]\leq C\E[Q(\wti{X}_{t\wedge\tau})] = C
  (\E[t\wedge\tau]+\E[Q(\wti{X}_0)])<\infty,\] and the same argument
  implies $\E[Z_T(\wti{X}_\tau)]<\infty$. So
  $\E[Z_T(\wti{X}_\tau)|\mathcal{F}_t]$ is a uniformly integrable
  martingale, and by convexity,
  $Z_T(\wti{X}_{t\wedge\tau})\leq\E[Z_T(\wti{X}_\tau)|\mathcal{F}_t]$. Therefore,
  \begin{equation*}
    -C|T-(t\wedge\tau)|\ \leq\ F(T)-G_T(\wti{X}_{t\wedge\tau},t\wedge\tau)
    \ \leq\ C|T-(t\wedge\tau)|+\E[Z_T(\wti{X}_\tau)|\Fc_t].
  \end{equation*}
  It follows that $\E[G_T(\wti{X}_{t\wedge\tau},t\wedge\tau)]
  \to\E[G_T(\wti{X}_{\tau},\tau)]$ as $t\to\infty$. On the other hand,
  \begin{equation*}
    \E\left[ H_T(\wti{X}_{\tau})\right]
    \ =\ \E\left[\int^{T}_{T\wedge R(\wti{X}_{\tau})}\Big[
      M(\wti{X}_{\tau},s)-f(s)\Big]\dif s\right]
    + \E\left[Z_T(\wti{X}_{\tau})\right]\ <\ \infty.
  \end{equation*}
  The same arguments hold when $\tau$ is replaced by $\tds$, and
  then we have
  \begin{equation*}
    \E\left[ H_T(\wti{X}_{\tau})\right]\ =\ \E\left[ H_T(\wti{X}_{\tds})\right]
    \hspace{15pt}\text{ and }\hspace{15pt}
    \E\left[ Z_T(\wti{X}_{\tau})\right]\ =\ \E\left[ Z_T(\wti{X}_{\tds})\right].
  \end{equation*}
  In addition, by Lemma~\ref{lem:GTmartsubmart}, we have,
  \begin{equation*}
    \E \left[G_T(\wti{X}_{T\wedge\tds},T\wedge\tds)\right]
    \geq\ \E \left[G_T(\wti{X}_{T\wedge\tau},T\wedge\tau)\right].
  \end{equation*}
  Combining the results above with \eqref{eq:GTineq}, we have
  \begin{equation}\label{eq:Fineq}
    \begin{split}
      \E F(\tau)&
       \leq \E\big[G_T(\wti{X}_\tau,\tau)+H_T(\wti{X}_\tau)\big]\\
      &= \E\big[G_T(\wti{X}_{T\wedge\tau},T\wedge\tau)
      +H_T(\wti{X}_\tau)\big]+\E\big[G_T(\wti{X}_\tau,\tau)
      -G_T(\wti{X}_{T\wedge\tau},T\wedge\tau)\big]\\
      &\leq \E\big[G_T(\wti{X}_{T\wedge\tds},T\wedge\tds)
      +H_T(\wti{X}_{\tds})\big]\\
      & \hspace{30pt} {}+\E\big[G_T(\wti{X}_\tau,\tau)
      -G_T(\wti{X}_{T\wedge\tau},T\wedge\tau)
      \big]\\
      &=
      \E\big[G_T(\wti{X}_{\tds},\tds)+H_T(\wti{X}_{\tds})\big]
      +\E\big[G_T(\wti{X}_\tau,\tau)-G_T(\wti{X}_{T\wedge\tau},
      T\wedge\tau)\big]\\
      &\hspace{30pt} {}-\E\big[G_T(\wti{X}_{\tds},\tds)-
      G_T(\wti{X}_{T\wedge\tds},T\wedge\tds)\big]\\
      &= \E\big[F(\tds)\big]
      +\E\left[\int_{T\wedge\tau}^{T}M(\wti{X}_{T\wedge\tau},s)\dif s
        -\int_{\tau}^{T}M(\wti{X}_{\tau},s)\dif s
        +Z_T(\wti{X}_{T\wedge\tau})\right]\\
      &\hspace{30pt}
      -\E\left[\int_{T\wedge\tds}^{T}M(\wti{X}_{T\wedge\tds},s)\dif
        s -\int_{\tds}^{T}M(\wti{X}_{\tds},s)\dif s
        +Z_T(\wti{X}_{T\wedge\tds})\right]\\
      &= \E\big[F(\tds)\big]
      +\E\left[\iden_{[\tau>T]}\int_{T}^{\tau}M(\wti{X}_{\tau},s)\dif
        s \right]\\
      &\hspace{30pt}
      {}-\E\left[\iden_{[\tds>T]}
        \int_{T}^{\tds}M(\wti{X}_{\tds},s)\dif s
      \right]+\E\Big[Z_T(\wti{X}_{T\wedge\tau})-Z_T(\wti{X}_{T\wedge\tds})\Big].
    \end{split}
  \end{equation}
  Since $f\leq C$, we have
  \begin{equation*}
    \begin{split}
      0 \leq \E\left[\iden_{[\tau>T]}
        \int_{T}^{\tau}M(\wti{X}_{\tau},s)\dif s\right]
       &\leq C\E\big[\iden_{[\tau>T]}(\tau-T)\big]\\
      &= C\E\big[\tau-T\wedge\tau\big]\ \longrightarrow\ 0, \ \ \
      \text{ as }\ T\rightarrow\infty.
    \end{split}
  \end{equation*}
  Similarly,
  \begin{equation*}
    \lim_{T\rightarrow\infty}
    \E\left[\iden_{[\tds>T]}\int_{T}^{\tds}
      M(\wti{X}_{\tds},s)\dif s\right]\ =\ 0.
  \end{equation*}
  Now, by the fact that $\E
  [Q(\wti{X}_{T\wedge\tau})]=\E[T\wedge\tau] + \E[Q(\wti{X}_0)]$ and the convexity of
  $Q$, $Q(\wti{X}_{T\wedge\tau})\le
  \E[Q(\wti{X}_\tau)|\mathcal{F}_T]$, hence,
  $Q(\wti{X}_{t\wedge\tau})\rightarrow Q(\wti{X}_\tau)$ in
  $L^1$. Noting that $Z_T(\wti{X}_{T\wedge\tau})\leq
  CQ(\wti{X}_{T\wedge\tau})$ and $Z_T(\wti{X}_{T\wedge\tau})\to
  CQ(\wti{X}_{\tau})$ a.s. as $T\to\infty$, we have
  \begin{equation*}
    \lim_{T\to\infty}\E\left[Z_T(\wti{X}_{T\wedge\tau})\right]
    \ =\ C\E\left[Q(\wti{X}_\tau)\right]\ <\infty.
  \end{equation*}
  The same arguments hold when $\tau$ is replaced by $\tds$, and
  moreover, $\E[Q(\wti{X}_\tau)]=\E[Q(\wti{X}_{\tds})]$. Now, let
  $T$ go to infinity in \eqref{eq:Fineq}, and we have
  \begin{equation*}
    \E\big[F(\tau)\big]\ \leq\ \E\big[F(\tds)\big].
  \end{equation*}
  To observe that the result still holds when $f$ is unbounded,
  observe that we can apply the above argument to $f(t) \wedge N$, and
  $F_N(t) = \int_0^t f(s) \wedge N \dif s$ to get $\E\left[
    F_N(\tds) \right] \ge \E\left[ F_N(\tau)\right]$, and the
  conclusion follows on letting $N\to \infty$. \end{proof}

Now we turn to the proofs of Proposition~\ref{prop:GTineq} and
Lemma~\ref{lem:GTmartsubmart}.

\begin{proof}[Proof of Proposition~\ref{prop:GTineq}]
  If $(x,t)\in D$, i.e. $t>R(x)$,
  \begin{equation*}
    \begin{split}
      G_{T}(x,t)+H_T(x) & =\int_{R(x)}^{t}M(x,s)\dif s
      +F(R(x)) \\
      & \geq\int_{R(x)}^{t}f(s)\dif s+F(R(x))=F(t).
    \end{split}
  \end{equation*}
  If $(x,t)\notin D$, i.e. $t\leq R(x)$,
  \begin{equation*}
    \begin{split}
      G_{T}(x,t)+H_T(x) & =-\int_{t}^{R(x)}M(x,s)\dif
      s+F(R(x))\\
      & =-\int_{t}^{R(x)}f(s)\dif s+F(R(x))=F(t).
    \end{split}
  \end{equation*}
\end{proof}

\begin{proof}[Proof of Lemma~\ref{lem:GTmartsubmart}]
  For $s\leq t\leq T$, by \eqref{eq:qvTbddcond}, the Meyer-It\^o
  formula gives,
  \begin{equation*}
    Z_T(\wti{X}_t)-Z_T(\wti{X}_s)
    \ =\ \int_s^t Z'_T(\wti{X}_u)\dif \wti{X}_u+\int_{s}^{t}M(\wti{X}_u,T)\dif u.
  \end{equation*}
  By \eqref{eq:qvTbddcond} and the fact $f$ is bounded, it is easy to
  see that the family $(Z_T(\wti{X}_t);0\leq t\leq T)$ is uniformly
  integrable. By the Doob-Meyer decomposition theorem
  (e.g. \citet[Theorem 4.10, Chapter 1]{KaratzasShreve:91}), the first
  term on the right-hand side is a uniformly integrable martingale,
  \begin{equation*}
    \E\big[Z_T(\wti{X}_t)-Z_T(\wti{X}_s)\,\big|\,\mathcal{F}_s\big]
    \ =\ \int_{s}^{t}\E\big[M(\wti{X}_u,T)\,\big|\,\mathcal{F}_s\big]
    \dif u.
  \end{equation*}
  Then we have,
  \begin{equation*}
    \begin{split}
      &G_{T}(\wti{X}_s,s)-\E\left[G_T(\wti{X}_t,t)\big|
        \mathcal{F}_s\right]\\
      & \hspace{10pt} {}= \int_{t}^{T}\E\big[M(\wti{X}_t,u)\big|
      \mathcal{F}_s\big]\dif u
      +\int_{s}^{t}\E\big[M(\wti{X}_u,T)\,\big|
      \,\mathcal{F}_s\big]\dif u-\int_{s}^{T}M(\wti{X}_s,u)\dif u\\
      & \hspace{10pt} {}= \int_{t}^{T}\E\big[M(\wti{X}_t,u)\big|
      \mathcal{F}_s\big]\dif u-\int_{s}^{T-t+s}M(\wti{X}_s,u)\dif u\\
      &\hspace{90pt}
      {} +\int_{s}^{t}\E\big[M(\wti{X}_u,T)\,\big|\,\mathcal{F}_s\big]\dif
      u
      -\int_{T-t+s}^{T}M(\wti{X}_s,u)\dif u\\
      & \hspace{10pt} {}= \int_{t}^{T}\Big\{\E\big[M(\wti{X}_t,u)\,\big|
      \,\mathcal{F}_s\big]
      -M(\wti{X}_{s},u-(t-s))\Big\}\dif u\\
      &\hspace{90pt} {}+\int_{s}^{t}\Big\{\E\big[M(\wti{X}_u,T)\,\big|
      \,\mathcal{F}_s\big]-M(\wti{X}_{s},T-(t-u))\Big\}\dif u.
    \end{split}
  \end{equation*}
  Now, observe that $\wti{X}_t$ is a Markov process, so we can write $Y_t$ for
  an independent copy of $\wti{X}$, and $\bar{\sigma}_D$ for the corresponding
  hitting time of the reversed barrier, and write $\wti{X}^x$ for $\wti{X}$
  started at $x$.  We have, for $u\in(t,T]$:
  \begin{align}
    \E^{(x,u-(t-s))}\big[f(\tds)|\mathcal{F}_u\big] \leq &
    \iden_{[\tds\leq u]}f(u)+\iden_{[\tds>u]}
    \E^{(x,u-(t-s))}\big[f(\tds)|\mathcal{F}_u\big]\nonumber\\
    = & \iden_{[\tds\leq u]}f(u)+\iden_{[\tds>u]}
    \E^{(\wti{X}_{t-s}^{x},u)}\big[f(\bar{\sigma}_D)\big]
    \nonumber\\
    \leq & M(\wti{X}_{t-s}^{x},u).  \label{eq:Rprep}
  \end{align}
  Hence,
  \begin{equation}
    \label{eq:R1}
    \begin{split}
      \E\big[M(\wti{X}_t,u)\,\big|\,\mathcal{F}_s\big]
      \ &=\ \E^{\wti{X}_s}M(\wti{X}_{t-s},u)\\
      &\geq\ \E^{(\wti{X}_{s},u-(t-s))} \big[f(\wti{\tau}_D)\big]\ =\
      M(\wti{X}_s,u-(t-s)).
    \end{split}
  \end{equation}

  For $u\in(s,T)$, replacing $u$ by $T$ and $t$ by $u$ in
  \eqref{eq:Rprep} gives that
  \begin{equation*}
    \E^{(x,T-(u-s))}\big[f(\tds)\,|\,\mathcal{F}_T\big]
    \ \leq\ M(\wti{X}^{x}_{u-s}, T),
  \end{equation*}
  and hence,
  \begin{equation}\label{eq:R2}
    \E\big[M(\wti{X}_u,T)\,\big|\,\mathcal{F}_s\big]
    \ \geq\ M(\wti{X}_s,T-(u-s)).
  \end{equation}
  It follows that
  \begin{equation*}
    \begin{split}
      \int_{s}^{t}\Big\{ & \E\big[M(\wti{X}_u,T)\,\big|
      \,\mathcal{F}_s\big]-M(\wti{X}_{s},T-(t-u))\Big\}\dif u\\
      & = \int_{s}^{t}\Big\{\E\big[M(\wti{X}_u,T)\,\big|
      \,\mathcal{F}_s\big]-M(\wti{X}_{s},T-(u-s))\Big\}\dif u\ \geq 0.
    \end{split}
  \end{equation*}
  Therefore,
  \begin{equation*}
    G_{T}(\wti{X}_s,s)-\E\left[G_T(\wti{X}_t,t)|\mathcal{F}_s\right]\ \geq\ 0,
  \end{equation*}
  which implies \eqref{eq:GTsupmart}.

  On the other hand, as a part of \eqref{eq:Rprep},
  \begin{equation*}
    \iden_{[\tds>u]}\E^{(x,u-(t-s))}\big[f(\tds)\,|
    \,\mathcal{F}_u\big]\ =\ \iden_{[\tds>u]}
    \E^{(\wti{X}_{t-s}^{x},u)}\big[f(\bar{\sigma}_D)\big],
  \end{equation*}
  and on $\{u<\tds\}$ equality
  holds in the inequalities \eqref{eq:R1} and \eqref{eq:R2}. Thus,
  \eqref{eq:GTmart} follows.
\end{proof}

For bounded $f$, although the pathwise inequality in this section
$G_T(\wti{X}_t,t)+H_T(\wti{X}_t)$ $\geq F(t)$ holds for all $T,t>0$,
$G_T(\wti{X}_t,t)$ is a supermartingale only on $[0,T]$. For hedging
purposes, we would really like to know: can we find a \emph{global}
pathwise inequality $G^\ast_t+H^\ast(\wti{X}_t)\geq F(t)$, such that
$G^\ast_t$ is a supermartingale on $[0,\infty]$ and a martingale on
$[0,\tds]$? We now provide conditions where we can find such
$G^\ast$ and $H^\ast$.

We replace \eqref{eq:fassumption} by a stronger assumption: there
exists some $\alpha>1$, such that \begin{equation}
  \label{eq:sfassumption}
  \text{for $t$ sufficiently large, }
  \ C\ \geq\ f(t)\ \geq\ C-O(t^{-\alpha}).
\end{equation}
Under this assumption, it is easy to check there exists a $J(x,t)$
such that
\begin{equation}\label{eq:Jdefn}
  J(x,t)\ =\ \lim_{T\to\infty}\int_t^T
  \left[M(x,s)-f(s)\right]\dif s,
\end{equation}
then we define
\begin{equation}\label{eq:GHlim}
  \left\{
    \begin{aligned}
      G(x,t)\ &=\ \lim_{T\to\infty}G_T(x,t)
      \ =\  F(t)-J(x,t)-CQ(x);\\
      H(x)\ &=\ \lim_{T\to\infty}H_T(x)\ =\ J(x,R(x))+CQ(x).
    \end{aligned}\right.
\end{equation}
Letting $T\to\infty$ in \eqref{eq:GTineq},
\begin{equation}
  \label{eq:sGTineq}
  \begin{cases}
    \ G(x,t)+H(x)\ >\ F(t),\ \ \ &\text{if }\ t>R(x);\\
    \ G(x,t)+H(x)\ =\ F(t),&\text{if }\ t\leq R(x).
  \end{cases}
\end{equation}
By the monotone convergence theorem, for all $t>0$,
$\E[\int^T_t[M(\wti{X}_t,s)-f(s)]\dif s]\to \E[J(\wti{X}_t,t)]$ as
$T\to\infty$, and then by Scheff\'e's Lemma,
$\int^T_t[M(\wti{X}_t,s)-f(s)]\dif s \to J(\wti{X}_t,t)$ in $L^1$. On
the other hand, since $Z_T(\wti{X}_t)\to CQ(\wti{X}_t)$ in $L^1$,
\begin{equation*}
  G_T(\wti{X}_t,t)\ \xrightarrow{L^1}\ G(\wti{X}_t,t)
  \ \text{ and }
  \ H_T(\wti{X}_t)\ \xrightarrow{L^1}\ H(\wti{X}_t).
\end{equation*}
It follows that the process $\left(G(\wti{X}_t,t);t\geq0\right)$ is a
supermartingale and the process $\left(G(\wti{X}_{t\wedge\tds},t\wedge\tds);t\ge
  0\right)$ is a martingale (since the conditional expectation, as an
operator, is continuous in $L^p$ for $p\geq1$). We then can show as
before that (if $\tau$, $\tds$ are integrable),
\begin{equation*}
  \begin{split}
    \E\left[F(\tau)\right]\ &\leq\ \E\left[G(\wti{X}_\tau,\tau)
      +H(\wti{X}_\tau)\right]\\
    &\leq\ \E\left[ G(\wti{X}_{\tds},\tds)
      +H(\wti{X}_{\tds})\right]\ =\ \E\left[F(\tds)\right].
  \end{split}
\end{equation*}

An example where \eqref{eq:sfassumption} holds is the call-type
payoff: $F(t)=(t-K)_+$. We see that for $t>K$, the left derivative
$f(t) = 1$, and hence
\begin{equation*}
  J(x,t) = \int_t^K [M(x,s)-f(s)]\dif s,
\end{equation*}
we then repeat all arguments above to obtain the pathwise inequality
and the optimality result. On the other hand, the condition will fail
if \eg{} $F(t) = t^2$. Recall that a volatility swap corresponds to
the choice of $F(t) = \sqrt{t}$, and that we can consider concave
functions by taking $-F(t)$. This causes difficulties since $-F$ is
not increasing, nor can we make it increasing by considering $-F(t) +
\alpha t$, for some $\alpha >0$, since $F'(t) \to \infty$ as $t \to
\infty$. However $F(t)$ can be approximated from both above and below
by functions $F^1(t)$ and $F^2(t)$ which are concave, and have bounded
derivatives. (In the case of the upper approximation, we have $F^1(0)
>0$, but this can be made arbitrarily small). An optimality result
will then follow in this setting --- note also that the condition
\eqref{eq:sfassumption} will be satisfied in this case.

\subsection{Superhedging options on weighted realised variance}

We now return to the financial context described in
Section~\ref{sec:financial-motivation}. Our aim is to use the construction we
produced for the proof of optimality in Section~\ref{sec:optim-via-pathw} to provide a
model-independent hedging strategy for derivatives which are convex functions
of weighted realised variance. We will suppose initially that our options are
not forward starting, so $\nu = \delta_{S_0}$.

We now define $\tds$ as the embedding of $\mu$ for the diffusion $\wti{X}$, and
define functions: $G, H, J$, and $Q$ as in the previous section (so
\eqref{eq:sfassumption} holds). Our aim is to
use \eqref{eq:sGTineq}, which now reads: 
\begin{equation}
  \label{eq:GHineq}
  G(X_{A_t},t)+H(X_{A_t})\ =\ G(\wti{X}_t,t)+H(\wti{X}_t)
  \ \geq\ F(t)\ =\ F\left(\int_{0}^{A_t}w(X_s)\sigma_s^2\dif s\right),
\end{equation}
to construct a super-replicating portfolio. We shall first show that
we can construct a trading strategy that super-replicates the
$G(\wti{X}_t,t)$ portion of the portfolio. Then we argue that we are
able, using a portfolio of calls, puts, cash and the underlying, to
replicate the payoff $H(X_T)$.

Since $(G(\wti{X}_t,t))_{t\geq0}$ is a supermartingale, we do
not expect to be able to replicate this in a completely self-financing
manner. However, by the Doob-Meyer decomposition theorem, and the
martingale representation theorem, we can certainly find some process
$(\wti{\phi}_t)_{t\geq0}$ such that: 
\begin{equation*}
  G(\wti{X}_t,t)\ \leq\ G(\wti{X}_0,0)
  +\int_0^t\wti{\phi}_s\dif\wti{X}_s
\end{equation*}
and such that there is equality at $t=\tds$. Moreover, since
$(G(\wti{X}_{t\wedge\tds},t\wedge\tds))_{t\geq0}$ is a
martingale, and $G$ is of $\mathbb{C}^{2,1}$ class in $D$ (since
$M(x,t)$ is), we have:
\begin{equation*}
  G(\wti{X}_{t\wedge\tds},t\wedge\tds)
  =  G(\wti{X}_0,0)+\int_{0}^{t\wedge\tds}
  \pian{G}{x}(\wti{X}_{s\wedge\tds},s\wedge\tds)\dif\wti{X}_s.
\end{equation*}
More generally, we would not expect $\partial G/\partial x$ to exist
everywhere in $D^\complement$, however, if for example left and right
derivatives exist, then we could choose
\begin{equation*}
  \wti{\phi}_t\ \in\ \left[\pian{G}{x}(x-,t),
    \pian{G}{x}(x+,t)\right]
\end{equation*}
as our holding of the risky asset
.

It follows then that we can identify a process
$\left(\wti{\phi}_t;t\geq0\right)$ with \begin{equation*}
  G(\wti{X}_{\tau_t},\tau_t)\ \leq\ G(\wti{X}_0,0)
  +\int_{0}^{\tau_t}\wti{\phi}_s\dif\wti{X}_s \ =\
  G(X_0,0)+\int_0^t\wti{\phi}_{\tau_s}\dif X_s,
\end{equation*}
where we have used e.g. \citet[Proposition V.1.4]{RevuzYor:99}.
Finally, writing $\phi_t=\wti{\phi}_{\tau_t}$, then
\begin{equation}\label{eq:Gineq}
  G(X_{t},\tau_t)\ \leq\ G(X_0,0)+\int_0^t\phi_s\dif X_s
  \ =\ G(X_0,0)+\int_0^t\phi_s \dif\left(B^{-1}_s S_s\right).
\end{equation}
If we consider the self-financing portfolio which consists of holding
$\phi_s B^{-1}_T$ units of the risky asset, and an initial investment
of $G(X_0,0)B^{-1}_T-\phi_0S_0B^{-1}_T$ in the risk-free asset, this
has value $V_t$ at time $t$, where $\dif \left(B_{t}^{-1}V_t\right)
=B^{-1}_T\phi_t\,\dif\left(B^{-1}_tS_t\right)$ and $V_0\ =\
G(X_0,0)B_T^{-1}$, and therefore
\begin{equation*}
  V_T\ =\ B_T\left(V_0B^{-1}_0+\int_0^T B_T^{-1}\phi_s
    \dif\left(B_{s}^{-1}S_s\right)\right)
  \ =\ G(X_0,0)+\int_0^T\phi_s\dif X_s.
\end{equation*}

We now turn to the $H(X_T)$ component in \eqref{eq:GHineq}. If $H(x)$
can be written as the difference of two convex functions (so in
particular, $H''(\dif K)$ is a well defined signed measure) we can
write: 
\begin{equation*}
  \begin{split}
    H(x) & = H(S_0)+H'_+(
    S_0)(x-S_0)+\int_{(S_0,\infty)}(x-K)_+ H''(\dif K)\\
    & \qquad {}+\int_{(0,S_0]}(K-x)_+ H''(\dif K).
  \end{split}
\end{equation*}
Taking $x=X_T=B_{T}^{-1}S_T$ we get:
\begin{equation*}\begin{split}
    H(X_T) & =  H(S_0)+H'_+( S_0)(B_{T}^{-1}S_T-S_0)
    +B_{T}^{-1}\int_{(S_0,\infty)}(S_T-B_T K)_+
    H''(\dif K)\\
    & \qquad {}+B_{T}^{-1}\int_{(0,S_0]}(B_T K-S_T)_+ H''(\dif K).
  \end{split}
\end{equation*}

This implies that the payoff $H(X_T)$ can be replicated at time $T$ by
`holding' a portfolio of: \begin{equation}
  \label{eq:Hport}
  \begin{split}
    &B_{T}^{-1}\big[H(S_0)-S_0H'_+(S_0)\big]\ \text{ in cash;}\\
    &B_{T}^{-1}H'_+(S_0)\ \,\,\text{ units of the asset;}\\
    &B_{T}^{-1}H''(\dif K)\ \text{ units of the call with strike }
    B_T K\text{ for }K\in(S_0,\infty);\\
    &B_{T}^{-1}H''(\dif K)\ \text{ units of the put with strike }
    B_T K\text{ for }K\in(0,S_0],\\
  \end{split}
\end{equation}
where the final two terms should be interpreted appropriately. In
practice, the function $H(\cdot)$ can typically be approximated by a
piecewise linear function, where the `kinks' in the function
correspond to traded strikes of calls or puts, in which case the
number of units of each option to hold is determined by the change in
the gradient at the relevant strike. The initial cost of setting up
such a portfolio is well defined provided the integrability condition:
\begin{equation}
  \label{eq:Hddcond}
  \int_{(0,S_0]}P(B_T K)|H''|(\dif K)
  +\int_{(S_0,\infty)} C(B_T K)|H''|(\dif K)\ <\ \infty,
\end{equation}
holds, where $|H''|(\dif K)$ is the total variation of the
signed measure $H''(\dif K)$. We therefore shall make the following
assumption:

\begin{assumption}
  \label{ass:Hass} The payoff $H(X_T)$ can be replicated using a
  suitable portfolio of call and put options, cash and the underlying,
  with a finite price at time $0$.
\end{assumption}

We can therefore combine these to get the following theorem:
\begin{thm}
  \label{thm:rHedge}
  Suppose Assumptions~\ref{ass:Price},~\ref{ass:Calls} and \ref{ass:Hass} hold, and suppose
  $F(\cdot)$ is a convex, increasing function with $F(0)=0$ and the left
  derivative $f(t):=F'_-(t)$ satisfies \eqref{eq:sfassumption}. Let $M(x,t)$
  and $J(x,t)$ be given by \eqref{eq:Mdefn} and \eqref{eq:Jdefn}, and are
  determined by the solution to {\bf{SEP}}${}^*(x w(x)^{-1/2},\delta_{S_0},\mu)$,
  where $\mu$ is determined by \eqref{eq:BL} and $w \in \mathcal{D}$.  We also
  define $Q$ after \eqref{eq:Qdefn}, such that \eqref{eq:qvTbddcond} holds and
  then the functions $G$ and $H$ are given by \eqref{eq:GHlim}.

  Then there exists an arbitrage if the price of an option with payoff
  $F(RV^w_T)$ is strictly greater than
  \begin{equation}
    \label{eq:Fub}
    \begin{split}
      B_T^{-1}\Big[G(S_0,0)+H(S_0) & +\int_{(S_0,\infty)}C(B_T
        K)H''(\dif K)\\
        & \qquad {}+\int_{(0,S_0]}P(B_T K)H''(\dif K)\Big].
    \end{split}
  \end{equation}
  Moreover, this bound is optimal in the sense that there exists a
  model which is free of arbitrage, under which the bound can be
  attained, and the arbitrage strategy can be chosen independent of
  the model.
\end{thm}

\begin{proof}

  According to the arguments above, our superhedge of the variance
  option can be described as the combination of a static portfolio
  \eqref{eq:Hport} and a self-financing dynamic portfolio which
  consists of an additional $B^{-1}_T\psi_t$ units of the risky asset
  and an additional initial cash holding of
  $B^{-1}_T\left(G(S_0,0)-\psi_0S_0\right)$. In the case where
  $G(x,t)$ is sufficiently differentiable, we can identify the process
  $\psi_t=\wti{\psi}_{\tau_t}$ by
  \begin{equation*}
    \psi_t = \pian{G}{x}(X_{t},\tau_t).
  \end{equation*}
  We observe that this strategy is independent of the true model. It
  is easy to see that the total initial investment of this superhedge
  is given be \eqref{eq:Fub}.

  In the case where $G$ is not sufficiently differentiable, we first
  observe that $G(x,t)$ is continuous: note that $Q(x)$ and $F(t)$ are
  trivially so by \eqref{eq:GHlim}, and $M(x,t)$ is continuous in
  $D$, and in $D^\complement$, and additionally at jumps of $R(x)$: it follows
  that $M(x_n,t) \to  M(x,t)$ as $x_n \to x$ except possibly at a set
  of Lebesgue measure zero, and hence $G(x,t)$ is continuous.

  Now consider $G(x,t)$ on a bounded open set of the form $\mathcal{O} =
  (y_0,y_1) \times (t_0,t_1)$. By continuity, $G$ can be approximated
  uniformly on the boundary $\partial \mathcal{O}$ (or more
  relevantly, on the boundary where $t>t_0$, $\partial \mathcal{O}_+$)
  by a smooth function. Specifically, for fixed $\eps >0$, there
  exists a function $G_\eps(x,t)$ such that $G(x,t) + \eps \ge
  G_{\eps}(x,t) \ge G(x,t)$ on $\partial \mathcal{O}_+$. Moreover, the
  function 
  \begin{equation}\label{eq:Gepsdefn}
    G_{\eps}(x,t) = \E^{(x,t)} \left[
      G_{\eps}(\wti{X}_{\tau_{\partial \mathcal{O}_+}}, \tau_{\partial
        \mathcal{O}_+}) \right]
  \end{equation}
  is $C^{2,1}$ and a martingale on $\bar{\mathcal{O}}$, and so
  \begin{equation*}
    G_{\eps}(\wti{X}_{\tau_{\partial \mathcal{O}_+}},\tau_{\partial
      \mathcal{O}_+}) = G_{\eps}(\wti{X}_0,t_0) +\int_{t_0}^{\tau_{\partial
        \mathcal{O}_+}}\pd{G_{\eps}}{x}\dif\wti{X}_s.
  \end{equation*}
  Since $G$ is a supermartingale, for $(x,t) \in \mathcal{O}$, from
  \eqref{eq:Gepsdefn} we have $G(x,t) \ge G_{\eps}(x,t) - \eps$.

  Now observe that we can choose a countable sequence of such sets
  $\mathcal{O}_1, \mathcal{O}_2, \ldots$ with each set centred at the
  exit point of the previous set, and such that any continuous path is
  guaranteed to pass through only finitely many such sets on a finite
  time interval. For any fixed $\delta > 0$, we can take a sequence of
  strictly positive $\eps_1, \eps_2, \ldots$ such that $\sum_{i}
  \eps_i = \delta$, and apply the arguments above to generate a
  sequence of functions $G_{\eps_i}(x,t)$ on $\mathcal{O}_i$. It
  follows that, given $\delta > 0$, we can always find a function
  $\wti{\psi}_t$ such that
  \begin{equation*}
    \delta + G(\wti{X}_0,0) + \int_0^t \wti{\psi}_t \, \dif \wti{X}_s \ge
    G(\wti{X}_t,t).
  \end{equation*}
  Since $\delta$ was arbitrary, whenever the price of an option is
  strictly greater than \eqref{eq:Fub} , we can choose $\delta$
  sufficiently small that the arbitrage still works. Finally, we
  observe that at any time $t \in [0,T]$, the arbitrage strategy is
  worth at least $F(\tau_t) \ge F(0)$, so the strategy is bounded
  below, and hence admissible.

  To see that this is the best possible upper bound, we need to show
  that there is a model which satisfies Assumption~\ref{ass:Price},
  has law $\mu$ under $\Q$ at time $T$, and such that the superhedge
  is actually a hedge. But consider the stopping time $\tds$ for the
  process $\wti{X}_t$. Define the process $\left(X_t;0\leq t\leq
    T\right)$ by
  \begin{equation*}
    X_t=  \wti{X}_{\frac{t}{T-t}\wedge\tds},
    \ \ \text{for }\ t\in[0,T],
  \end{equation*}
  and then $X_t$ satisfies the stochastic differential equation \[\dif
  X_s=\hat{\sigma}_s X_s w(X_s)^{-1/2} \dif W_s = \sigma_s^2 X_s \dif W_s\] with
  the choice of
  \begin{equation}\label{eq:choicesigma}
    \hat{\sigma}^2_s
    = \dfrac{T}{(T-s)^2}\iden_{[\frac{s}{T-s}<\tds]}, \quad \sigma^2_s  =
    \hat{\sigma}_s^2 w(X_s)^{-1/2}.
  \end{equation}
  Since $\tds<\infty$, a.s., then $X_T=\wti{X}_{\tds}$, and
  \[\tau_T= \int_0^T
  w(X_s) \sigma_s^2 \ds = \int_0^T
  \dfrac{T}{(T-s)^2}\iden_{[\frac{s}{T-s}<\tds]} = \tds.\] Hence $S_t=X_t B_t$
  is a price process satisfying Assumption~\ref{ass:Price} with
  \begin{equation*}
    F\left(\int_0^Tw(X_s) \sigma^2_t\dif t\right)\ =\ F(\tds).
  \end{equation*}
  Finally, it follows that at time $T$, the value of the hedging
  portfolio exactly equals the payoff of the option. 
\end{proof}

\begin{rmk}\label{rmk:ForwardStart}
  The above result assumes that the option payoff depends on the
  realised weighted variation computed between time $0$ and a fixed
  time $T$. In some situations, {\it forward-starting} versions of
  these derivatives may be traded. Here, one is interested in the
  payoff of an option written on the variation observed between a
  fixed time $T_0 >0$ and the maturity date $T_1$: $\int_{T_0}^{T_1}
  w(X_t)\sigma^2_t \, \dt$. If one observes traded options at both
  $T_0$ and $T_1$, these again imply the (hypothesised, risk-neutral)
  distributions at times $T_0$, and $T_1$, and it is reasonable to
  suppose that the upper bound on the price of an option (for
  suitable, convex $F(\cdot)$) should correspond to the solution of
  \SEP{} determined above. Let $G$ and $H$ be the functions derived
  above. The question remains as to how one includes the additional
  information at time $T_0$ in the hedging strategy. (For clarity, we
  suppose $B_t = 1$ for all $t\ge 0$.)

  In order to have the correct hedge for $t \in [T_0,T_1]$,
  we need a portfolio of call options maturing at time $T_1$ with
  payoff $H(X_{T_1})$. In addition to the payoff at maturity, we need
  a dynamic portfolio worth (at least) $G(X_t,\tau_t)$, where now
  $\tau_t = \int_{T_0}^t \lambda(X_s) \sigma_s^2 \, \ds$ ---
  specifically, recalling $F(0)=0$, and Proposition~\ref{prop:GTineq},
  we should have $G(X_{T_0},0) + H(X_{T_0}) = 0$. This implies that we
  need a portfolio of call options with maturity $T_0$ and with payoff
  $-H(X_{T_0})$. Under a similar assumption to
  Assumption~\ref{ass:Hass}, this is possible, and the resulting
  strategy will give a superhedge which is a hedge under the optimal
  model corresponding to the Rost embedding.
\end{rmk}

Strictly speaking, Theorem~\ref{thm:rHedge} is model-dependent: our
arbitrage strategy is specified in a way that is independent of the
exact model, but some of the underlying concepts --- specifically the
quadratic variation in the option payoff, and the stochastic integral
term that is implemented in the hedge both depend on an underlying
probability space, and it could therefore be argued that the
strategies are not truly model-independent. In the following remarks,
we briefly outline how one might relax this assumption.

\begin{rmk}
  In a similar manner to recent work of \citet{DavisOblojRaval:10}, we
  can formulate this result without any need for a probabilistic
  framework. The difficulty in treating the previous arguments on a
  purely pathwise basis is that we need to make sense of the
  stochastic integral term in \eqref{eq:Gineq}, and the quadratic
  variation in the option payoff. However, under mild assumptions on
  the paths of $S_t=B_t^{-1} X_t$, and a stronger assumption on $G$
  (specifically, that $G$ is $\mathbb{C}^{2,1}$)\footnote{This would
    appear to be a very strong assumption on $G$. However, along the
    lines of \cite{DavisOblojRaval:10}, it seems reasonable that the
    conclusions would hold in a milder sense; what seems harder is to
    both provide a set of conditions under which these conclusions
    hold, and which can be verified under relatively natural
    constraints on our modelling setup.}, we can recover a pathwise
  result, based on a version of It\^o's formula due to
  \citet{Follmer:81}.

  Suppose we fix a sequence of partitions $\pi_n = \{0 = t_0^n \le
  t_1^n \le t_2^n \le \cdots \le t_{n}^n=T\}$ of $[0,T]$, such that
  $\sup_{i \le n,}|t_{i}^n-t_{i-1}^n| \to 0$ as $n \to \infty$. Then
  we define the class $\mathcal{QV}$ of continuous, strictly positive
  paths $X_t$ such that
  \begin{equation}
    \label{eq:QVdefn}
    \sum_{i =1}^n \left(\frac{
        X_{t_{i}^n}-X_{t_{i-1}}^n}{X_{t_{i-1}}^n}\right)^2 \delta_{t_{i-1}^n} = \mu^n
    \to \mu \text{ where } \mu([0,t]) = \int_0^t \sigma_s^2 \, \ds
  \end{equation}
  for some bounded measurable function $\sigma_s:[0,T] \to \R_+$.
  Here $\delta_t$ is the Dirac measure at $t$, and the convergence is
  in the sense of weak convergence of measures as $n \to \infty$,
  possibly down a subsequence.

  Then, following the proof of the main theorem in \citet{Follmer:81},
  an application of Taylor's Theorem to the terms
  $G(X_{t_{i}^n},\tau_{t_i^n})-G(X_{t_{i-1}^n},\tau_{t_{i-1}^n})$,
  where $\tau_t = \int_0^t \lambda(X_s) \sigma_s^2 \, \ds$, gives
  \begin{align*}
    G(X_T,\tau_T) - G(X_0,0) = {} & \sum_{i =1}^n \pd{G}{x} \left(
      X_{t_i^n}-X_{t_{i-1}^n}\right) + \sum_{i =1}^n \pd{G}{t} \left(
      \tau_{t_i^n}-\tau_{t_{i-1}^n}\right) \\& {}+ \half \sum_{i =1}^n
    \pdd{2}{G}{x} \left( X_{t_i^n}-X_{t_{i-1}^n}\right)^2.
  \end{align*}
  It follows that, whenever $X_t$ is a path in $\mathcal{QV}$, then:
  \begin{align*}
    \sum_{i =1}^n \pd{G}{x} \left( X_{t_i^n}-X_{t_{i-1}^n}\right)
    \to {} & G(X_T,\tau_T) - G(X_0,0) \\
    & {} - \int_0^T \sigma_s^2 \left( \lambda(X_s)\pd{G}{t} + \half
      \pdd{2}{G}{x} X_s^{2}\right) \, \ds.
  \end{align*}
  Recall that $\sigma(x) = x \lambda^{-1/2}(x)$, and since $G \in
  \mathbb{C}^{2,1}$ and a supermartingale for $X_t$ where $X_t$ solves
  \eqref{eq:Xdefn}, the final integrand will be negative. We conclude
  that, in the limit as we trade more often, \emph{for any} $X_t \in
  \mathcal{QV}$, we will have a portfolio which superhedges. One could
  then recover the statement in the probabilistic setting by
  observing that, almost surely, a path from a model of the form
  described by Assumption~\ref{ass:Price} lies in $\mathcal{QV}$.
\end{rmk}

\begin{rmk}\label{rmk:uncertainvolatility}
  An alternative approach that is still within a more general,
  model-independent setting, but where we do not need to assume strong
  differentiability conditions on $G$, can be constructed using the
  uncertain volatility approach, originally introduced by
  \citet{Avellaneda:1995aa}. We base our presentation on the paper of
  \citet{possamai_robust_2013}.

  Let $\Omega = \left\{\omega \in C([0,T];(0,\infty)), \omega(0) = X_0
  \right\}$ be a path space, equipped with the uniform
  norm, $||\omega|| = \sup_{t \in [0,T]}|\omega(t)|$, and let
  $X_t(\omega) = \omega(t)$ be the canonical process. Let $\Pr_0$ be
  the probability measure on $\Omega$ such that $X_t$ is a standard
  geometric Brownian motion (\ie{} $\log{X_t}$ has quadratic variation
  $t$). Let $\mathbb{F}$ be the filtration generated by $X$.

  Let $\mathbb{H}_{loc}^e(\Pr_0,\mathbb{F})$ be the set of
  non-negative, $\mathbb{F}$-progressively measurable processes
  $\alpha_t$ such that $\exp\left\{\half \int_0^\cdot \alpha_s \, \ds
  \right\}$ is $\Pr_0$-locally integrable. Then for $\alpha \in
  \mathbb{H}_{loc}^e(\Pr_0,\mathbb{F})$ we can define
  \begin{equation*}
    X_t^\alpha = \exp
    \left\{
      \int_0^t \alpha_s^{1/2} \, \dif \log(X_s)  - \half \int_0^t
      \alpha_s \, \ds
    \right\}.
  \end{equation*}
  In particular, under $\Pr_0$, $X_t^\alpha$ has $\left< \log
    X\right>_t = \int_0^t \alpha_s \, \ds$. Then we can define a
  probability measure on $\Omega$ by $\Pr^\alpha(X_t \in
  A) = \Pr_0(X_t^\alpha \in A)$, or equivalently, $\Pr^\alpha = \Pr_0
  \circ (X_\cdot^\alpha)^{-1}$. It follows that there is a class of
  probability measures $\mathcal{P} = \left\{\Pr^\alpha: \alpha \in
    \mathbb{H}_{loc}^e(\Pr_0,\mathbb{F})\right\}$ on the space
  $(\Omega,\mathbb{F})$. We aim to produce conclusions which hold for
  all $\Pr^\alpha$, and we say that something holds
  $\mathcal{P}$-quasi surely (q.s.) if it holds $\Pr$-\as{} for all $\Pr \in
  \mathcal{P}$.

  We now have a filtered space $(\Omega,\mathbb{F})$, and a class of
  (non-dominated) probability measures $\mathcal{P}$ under which we
  can discuss trading strategies simultaneously. Observe that the
  variance process $\langle \log X \rangle_t$ can be defined pathwise
  on $\Omega$ using the results of \citet{karandikar_pathwise_1995}:
  set $a_0^n = 0$, and $a_{i+1}^n = \inf\{t \ge a_i^n : |\log(\omega(t)) -
  \log(\omega(a_i^n))| \ge 2^{-n}\}$, and consider the process
  $V_t(\omega)$ defined by
  \begin{equation} \label{eq:Vnolimit}
    V_t(\omega) =
    \lim_{n \to \infty} \sum_{i=0}^n \left(
      \log(\omega(a_{i+1}^n\wedge t))-\log(\omega(a_{i}^n\wedge t)) \right)
  \end{equation}
  if the limit exists, where the limit is taken in the sense of
  uniform convergence on $[0,T]$ and defined to be zero otherwise. The
  limit exists $\Pr^\alpha$-\as{} for each $\alpha$, and when the
  limit exists, the limit is $\mathcal{F}_t$ measurable and
  $\Pr^\alpha$-\as{} equal to $\langle \log X \rangle_t$. If we write
  \begin{equation*}
    \mathcal{N}^\mathcal{P} = \{ E \subset \Omega: \exists \wti{E} \in
    \mathcal{F} \ s.t.{}\  E \subseteq \wti{E},\  \Pr(E) = 0 \ \forall
    \ \Pr \in
    \mathcal{P}\},
  \end{equation*}
  then the set of $\omega$ for which the limit in \eqref{eq:Vnolimit}
  fails is an element of $\mathcal{N}^\mathcal{P}$. As a result, we can make the
  process $V_t$ adapted by considering the augmented filtration:
  \begin{equation*}
    \mathcal{F}_t' = \mathcal{F}_t \vee \mathcal{N}^\mathcal{P}, \quad
    \mathbb{F}' = \{\mathcal{F}_t', t \ge 0\}.
  \end{equation*}
  Under this augmented filtration, the process $V_t$ remains
  $\mathbb{F}'$-progressively measurable, and indeed is continuous;
  it follows that the trading strategy described in the proof of
  Theorem~\ref{thm:rHedge} can be constructed, giving a c\`adl\`ag
  process $\psi_t$ which is continuous except at the times when the
  process $(\omega(t),V_t(\omega))$ exits the sets $\mathcal{O}_1,
  \mathcal{O}_2, \ldots$. Using the construction of
  \citet{karandikar_pathwise_1995}, we can again define pathwise a
  process $I_t$, which agrees $\Pr^\alpha$-\as{} with the classical
  stochastic integral $J_t^\Pr = \int_0^t \psi_s \dif X_s$. (Observe
  however that we may need to work in a $\Pr^\alpha$-augmented
  filtration for this latter object to be defined). Since by
  Theorem~\ref{thm:rHedge} we have $J_T^\Pr \ge G(X_T,\langle \log X
  \rangle_T)-\eps$, it follows that $I_T \ge G(X_T,V_T) - \eps$
  $\mathcal{P}$-q.s., and therefore the strategy we describe makes
  sense in the uncertain volatility setting.

  The fact that we have a concrete characterisation of
  $\psi_t$ enables us to avoid much of the technical difficulties that
  arise in \cite{possamai_robust_2013} and related papers. However,
  our results are in one sense also not quite so strong: we only
  obtain a strategy which superhedges our payoff less some
  $\eps>0$. The results in \cite{possamai_robust_2013} suggest that
  this is unnecessary. However, our results are stronger in another
  direction: we do not require any integrability restriction on the
  payoff of the option under the class of models we consider --- this
  constraint is already embedded in our restriction to non-negative
  price processes.
  
\end{rmk}

\section{Numerical Results}

\label{sec:numerical-results}

\subsection{Numerical solution of the viscosity equation}

An important goal is to use the results of the previous sections to
find numerical bounds, and their associated option prices and hedging strategies,
corresponding to the solutions of Rost and Root. The hardest aspect of
this is finding the numerical solution to the viscosity equation
\eqref{eq:ViscSoln}, and its equivalent for the Root solution. The
solution to the Rost viscosity equation is roughly equivalent to
solving a parabolic PDE inside the continuation region, while outside
this region we know the solution will be equal to the initial boundary
condition.

The numerical solution is made harder by the fact that, particularly
in the case of the Rost solution, we expect the behaviour of the
barrier near the initial starting point to be very sensitive to any
discretisation: in the case where the starting measure is a point mass
at $X_0$, and the target measure also places mass continuously (say)
near $X_0$, then we are looking for a barrier function $R(x)$ with
$R(X_0) = 0$, and a positive, but non-zero probability that $R(X_t)>t$
for some small time $t$. According to the law of the iterated
logarithm, the behaviour of the stopped process
will be very sensitive to small changes in $R(\cdot)$. As a result, a
numerical method that can concentrate on this initial region would be
beneficial. On the other hand, the behaviour of the barrier at large
times is also of interest, although here we expect the numerics are
likely to be less sensitive to discretisation.

A second question concerns the convergence and stability of our
numerical methods. The theory behind the numerical approximation of
viscosity equations is fairly well understood --- dating back to the
methods of \citet{Barles:1991aa}. In this paper, we use the results of
\citet{Barles:2007aa}, which are suited to our purposes. Since we wish
to use a large range of time steps, and we look to have non-equal grid
point spacings, we will look to use an implicit method, in order to
provide unconditional stability of the numerical regime. The results
of \cite{Barles:2007aa} provide us with the necessary justification.

To outline the numerical method used, we consider a standard numerical
scheme, with $\vec{u}^n$ the (vector valued) approximation to $u(x,t)$
evaluated at $t=t^n$, and at spatial positions $\vec{x}$. We
approximate $\mathcal{L}u = \frac{\sigma(x)^2}{2} \pd{^2 u}{x^2}(x,t)$
using the Kushner approximation described by \cite{Barles:2007aa},
which ensures that the finite difference operator $L$ can be written
in the form $(L\vec{u})_i = \sum_{j} c_{j}(t^n,x_{i+j})
(u_{i+j}-u_i)$, where the $c_j$'s are non-negative and zero except on
some finite subset of $\mathbb{Z} \setminus \{0\}$. We also need to
assume that the measures $\mu$ and $\nu$ both have compact support
and the same mean, in which case $\vec{u}$ is constant and zero at the
endpoints of $\vec{x}$.

Then an implicit numerical scheme to solve \eqref{eq:ViscSoln} will
take $\vec{u}_0= u(\vec{x},0)$, and solve iteratively
\begin{equation}\label{eq:DiscretePDE}
  \frac{\vec{u}_{n+1}-\vec{u}_n}{t^{n+1}-t^n} = \max\{L\vec{u}_{n+1},0\}.
\end{equation}
The difficulty here arising from the fact that the maximisation
depends on the unknown $\vec{u}_{n+1}$.  We can rearrange this
expression, writing $\vec{z} = \vec{u}_{n+1}-\vec{u}_{n}$, and $\Delta
t^n = t^{n+1}-t^n$, to see that this is equivalent to the problem of
finding $\vec{z} \ge 0$ such that: $(I-\Delta t^n L) \vec{z} - \Delta
t^n L \vec{u}^n \ge 0$, and $\vec{z}^\intercal \left((I-\Delta t^n L)
  \vec{z} - \Delta t^n L \vec{u}^n\right) = 0$. This is a classical
linear complementarity problem (LCP), and may be hard to solve (or at least,
may involve many evaluations of the matrix multiplication inside the
maximisation), however, at this point we can exploit the fact that the
structure of the solution implies that $\vec{z}$ will be zero at
exactly the points where we are in the barrier. Since the barrier
should generally change relatively slowly, as an initial supposition,
it is likely that the spatial values where $\vec{z} = 0$ for the
previous time-step are likely to be roughly the same at the next
step. It follows that a numerical scheme for solving LCPs which
involves pivoting on a set of basis variables may be very efficient at
solving \eqref{eq:DiscretePDE}. The algorithm we will use for this
purpose is the Complementary Pivot (or Lemke's) algorithm. We refer to
\citet{Murty:1988aa} for details on the numerical implementation of
the Complementary Pivot algorithm. We note also that a similar method
can be used to justify implicit methods for the Root solution
(the case of explicit solutions being justified directly by the
results of \cite{OberhauserDosReis:13}).

\subsection{Analysis of numerical evidence}

Using the methods outlined above, we can analyse the solutions of Root
and Rost numerically. 
In general, we
consider $\nu = \delta_{S_0}$ and $\mu$ will be determined by assuming
that we observe prices of call options which are consistent with a
Heston market model. In general we will consider features of barriers
under Heston models since they permit relatively straightforward
computation of both call prices, and prices of variance options. In
what follows, we take our given prices to come from a Heston model
with parameters: $\rho = -0.65, v_0 = 0.04, \theta = 0.035, \kappa =
1.2, r = 0, \xi = 0.5, S_0 = 2$ (see \eqref{eq:Heston} for the meaning
of the parameters).

\begin{figure}[htb]
  \centering
  \begin{tabular}{@{}cc@{}}
    \includegraphics[width=.48\textwidth]{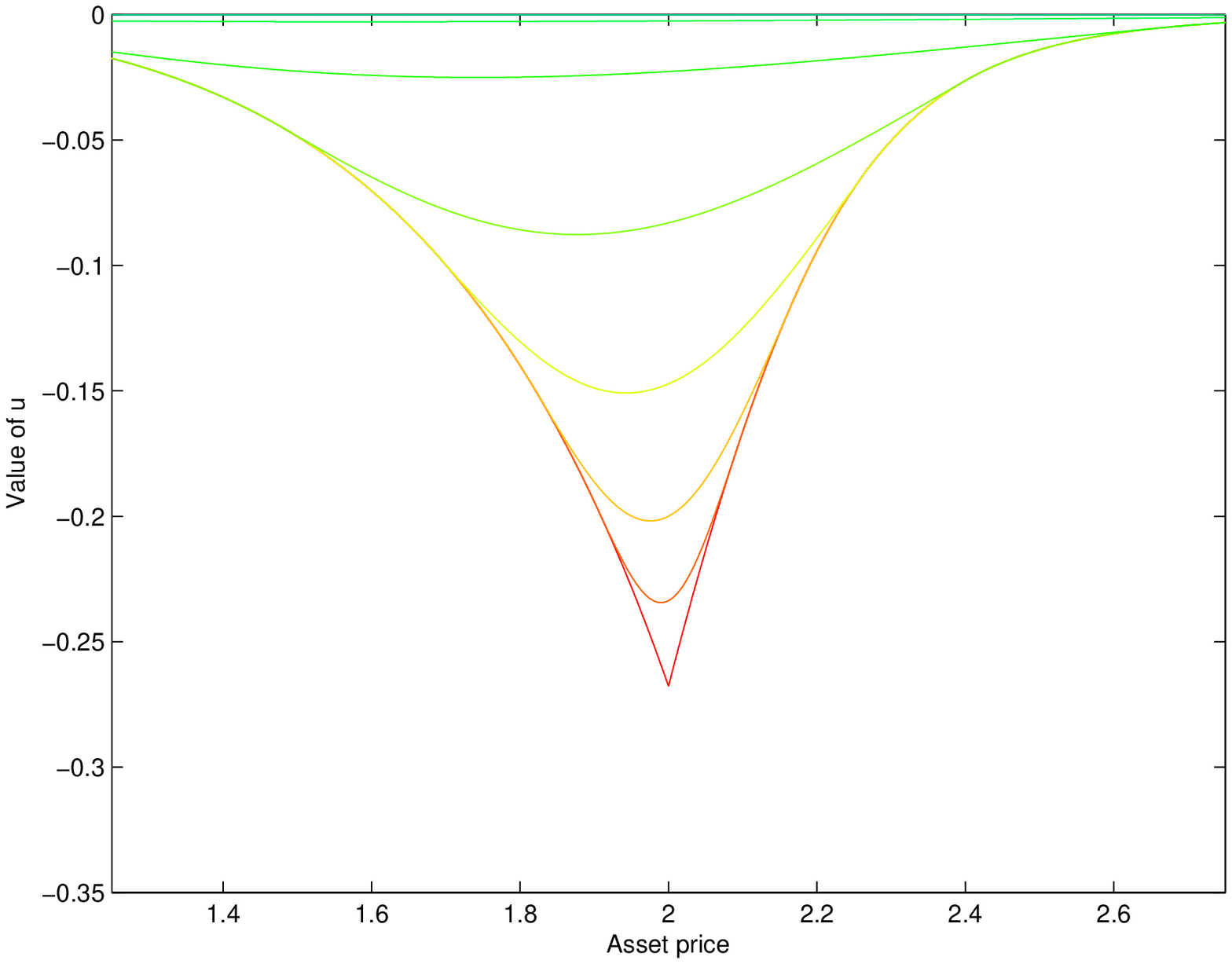} & 
    \includegraphics[width=.48\textwidth]{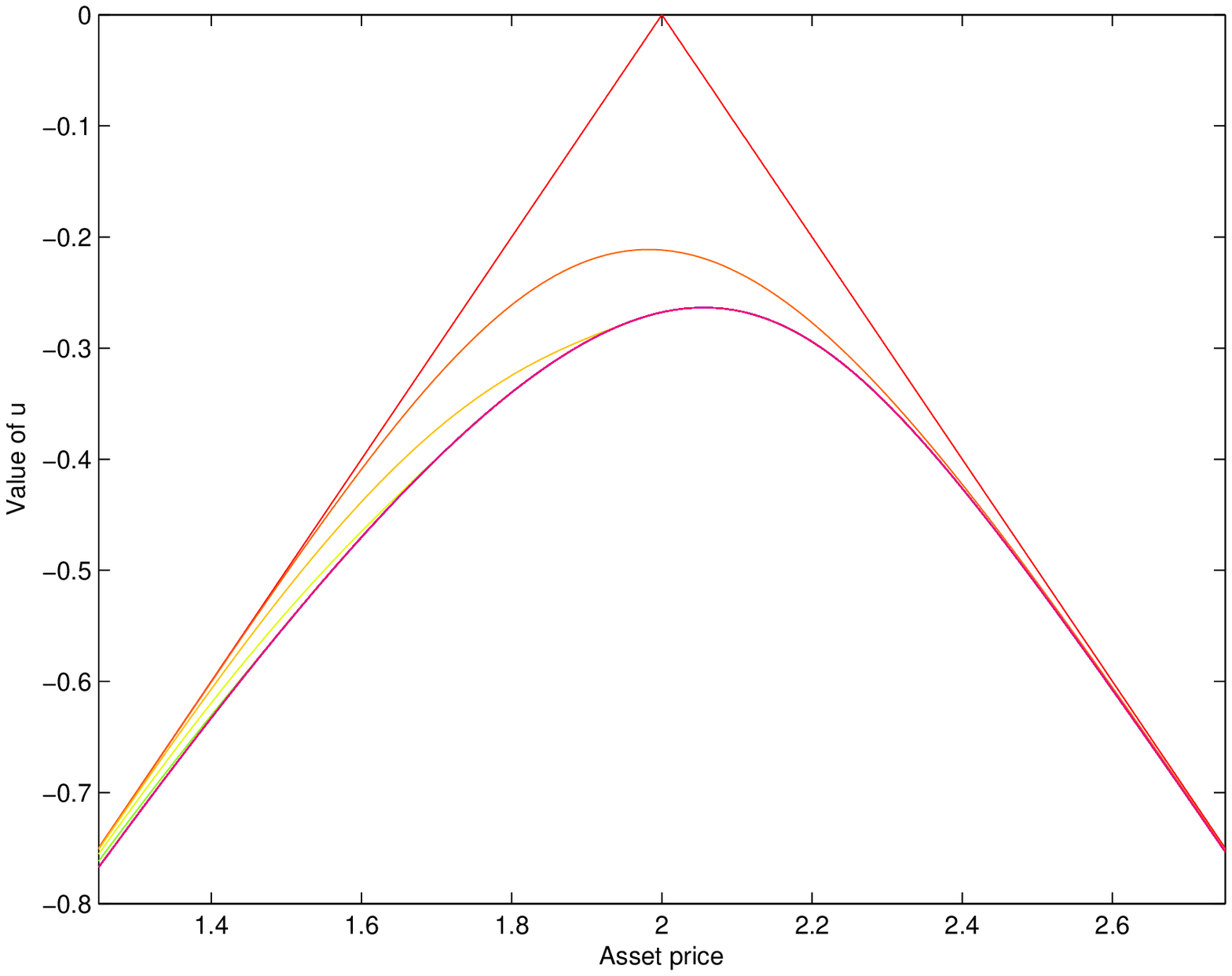} \\ 
    \includegraphics[width=.48\textwidth]{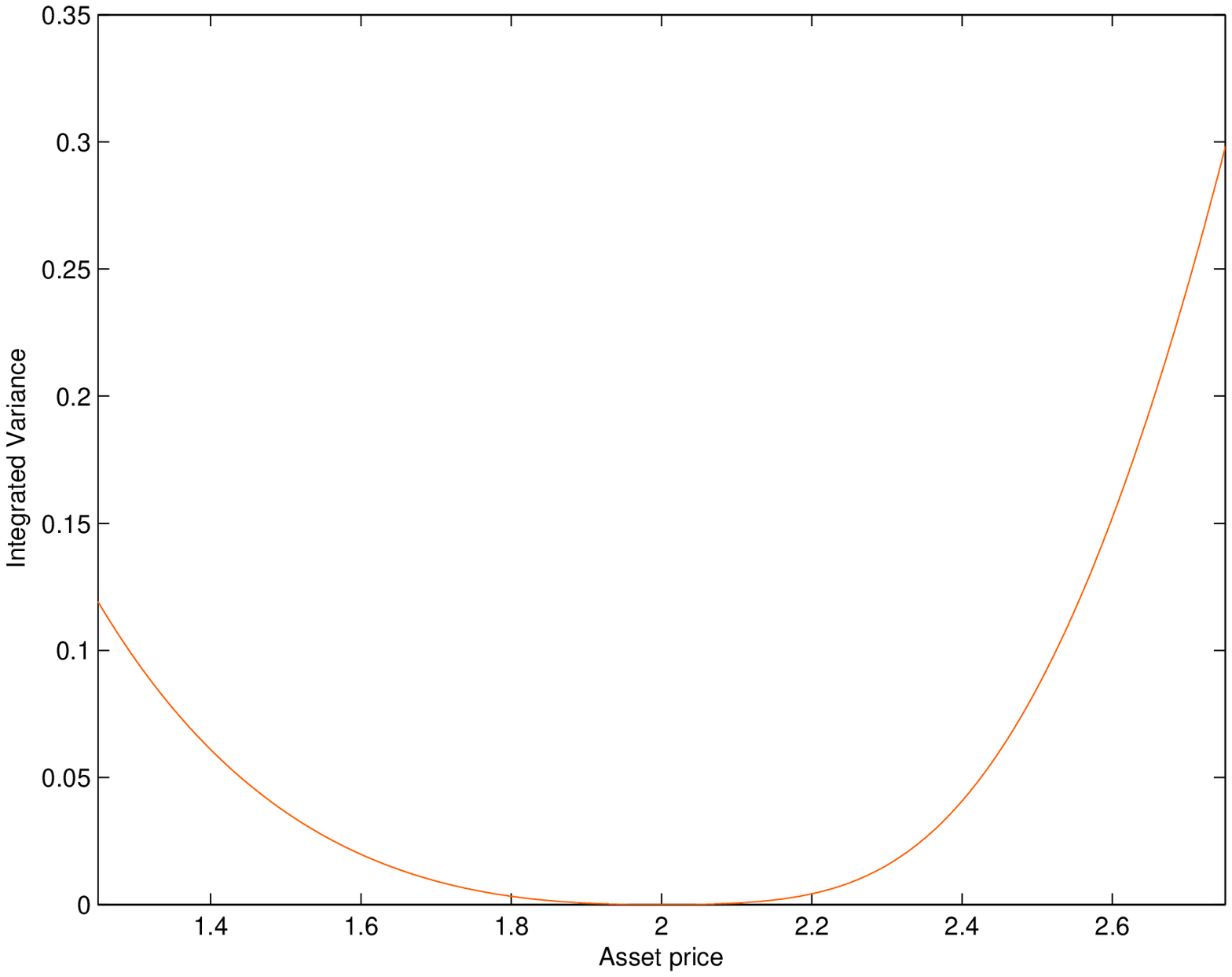} & 
    \includegraphics[width=.48\textwidth]{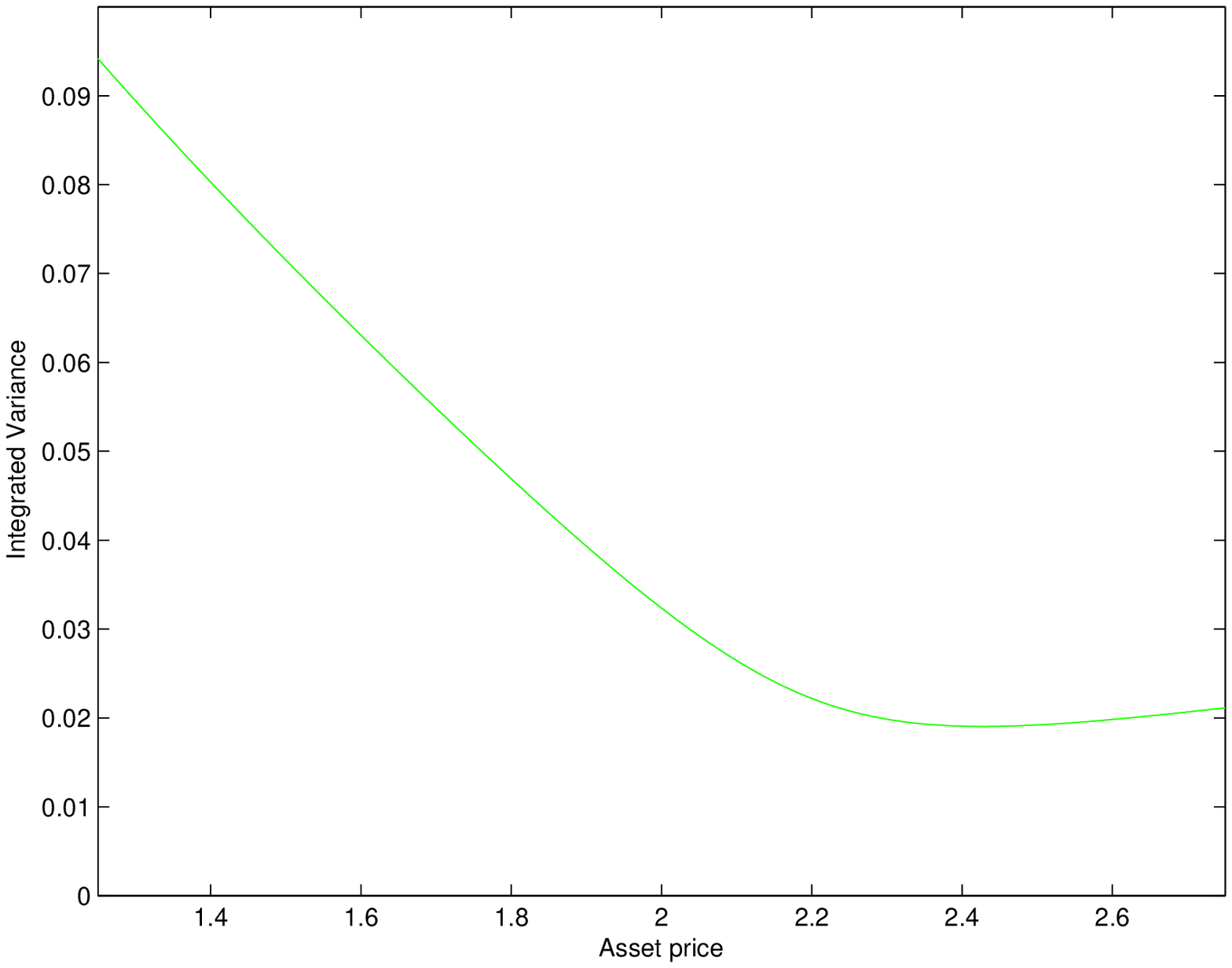} 
  \end{tabular}
  \caption{Plots of the function $u(x,t)$ (top) and the corresponding barriers (bottom) for the Rost (left) and Root (right) barriers.}
  \label{fig:basicbarriers}
\end{figure}

\begin{figure}[htb]
  \centering
  \begin{tabular}{@{}cc@{}}
    \includegraphics[width=.48\textwidth]{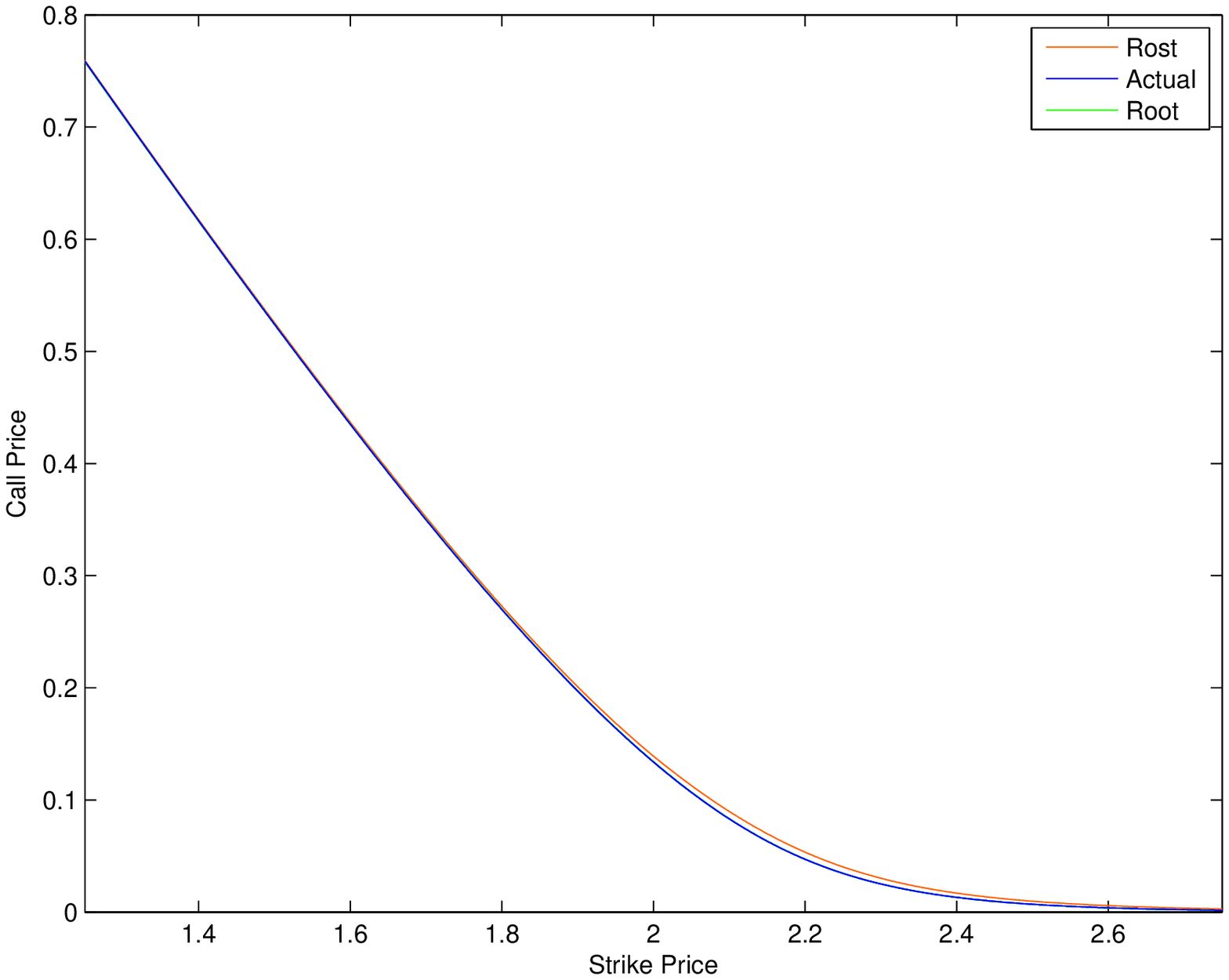}
    & 
    \includegraphics[width=.48\textwidth]{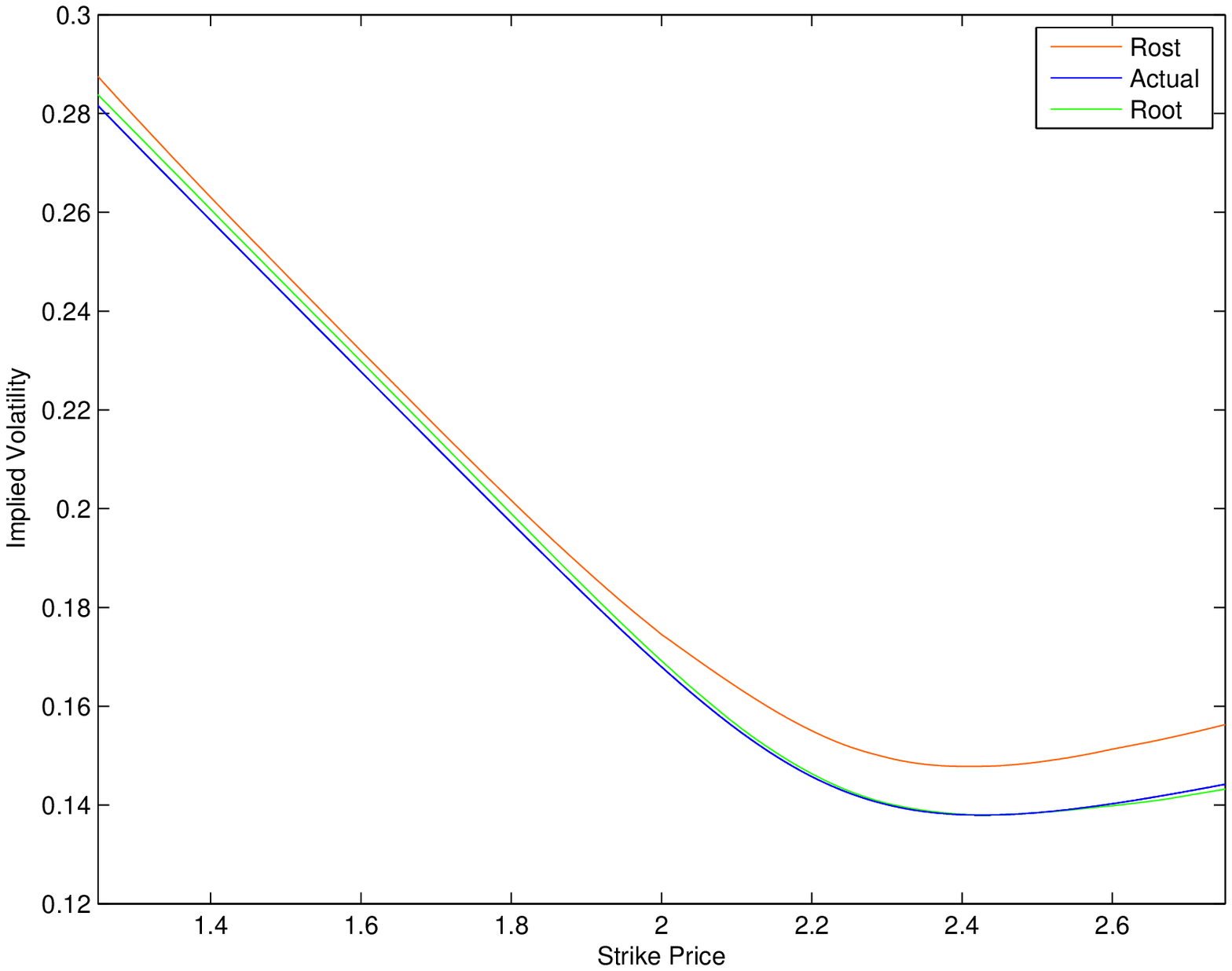}
  \end{tabular}
  \caption{The original call prices from which we obtained our
    barrier, and the empirical call prices obtained by simulation for
    the Rost and Root barriers (left); The implied volatility of the
    call prices (right). Note that numerically the Rost barrier proves
  harder to correctly compute/simulate.}
  \label{fig:EmpiricalImpliedVol}
\end{figure}
 Figure~\ref{fig:basicbarriers} shows the functions $u(x,t)$ in both the Rost and
Root solutions, and their corresponding barrier functions. We can
confirm that these functions do indeed embed the correct distributions
by simulation: we compute the distribution of a process stopped on
exit from the barrier and compute the corresponding call prices
empirically. In fact, it is more informative to plot the implied
volatility of the empirically obtained call prices. This is done in
Figure~\ref{fig:EmpiricalImpliedVol}.

\begin{figure}[htb]
  \centering
  \begin{tabular}{@{}cc@{}}
    \includegraphics[width=.48\textwidth]{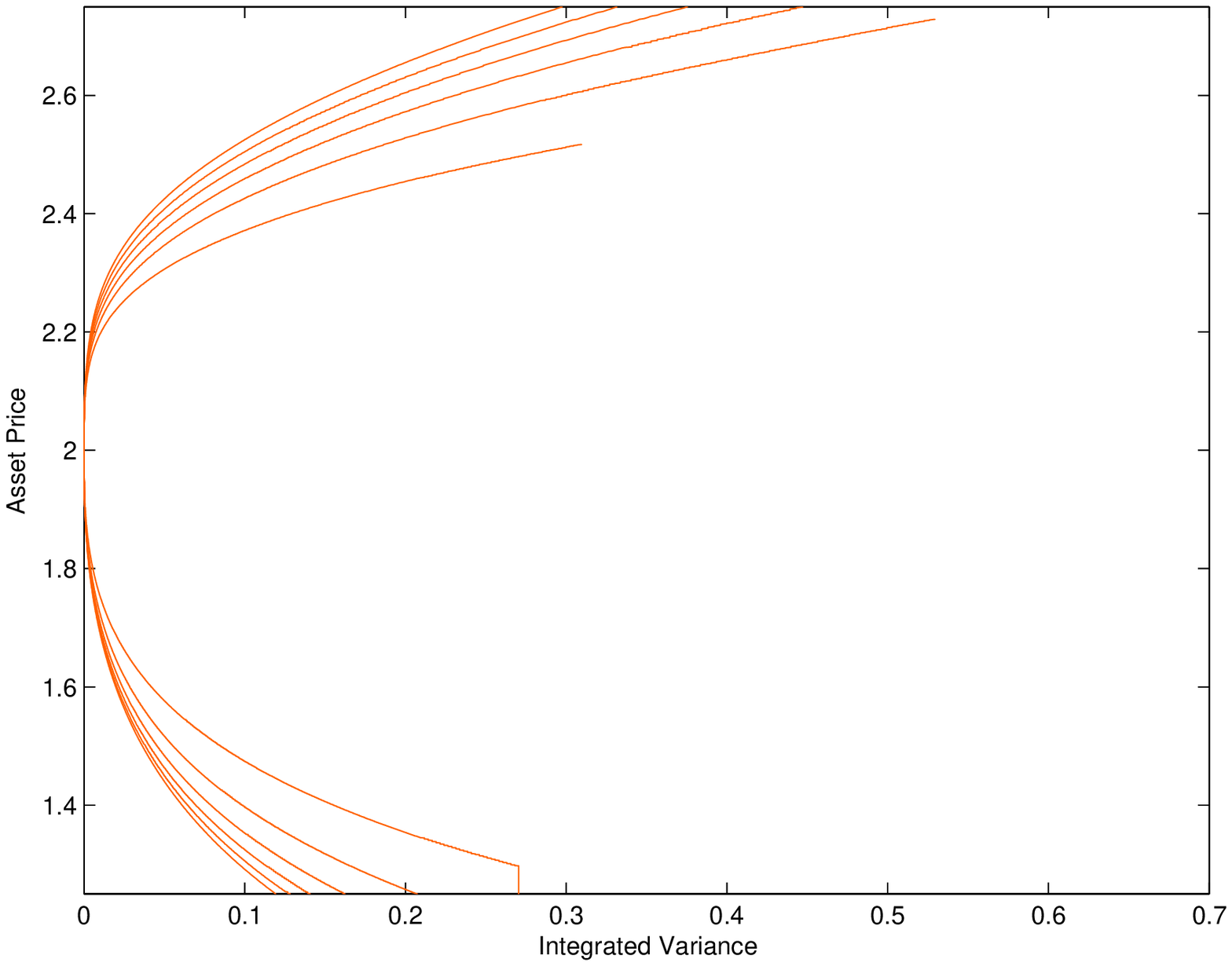}
    & 
    \includegraphics[width=.48\textwidth]{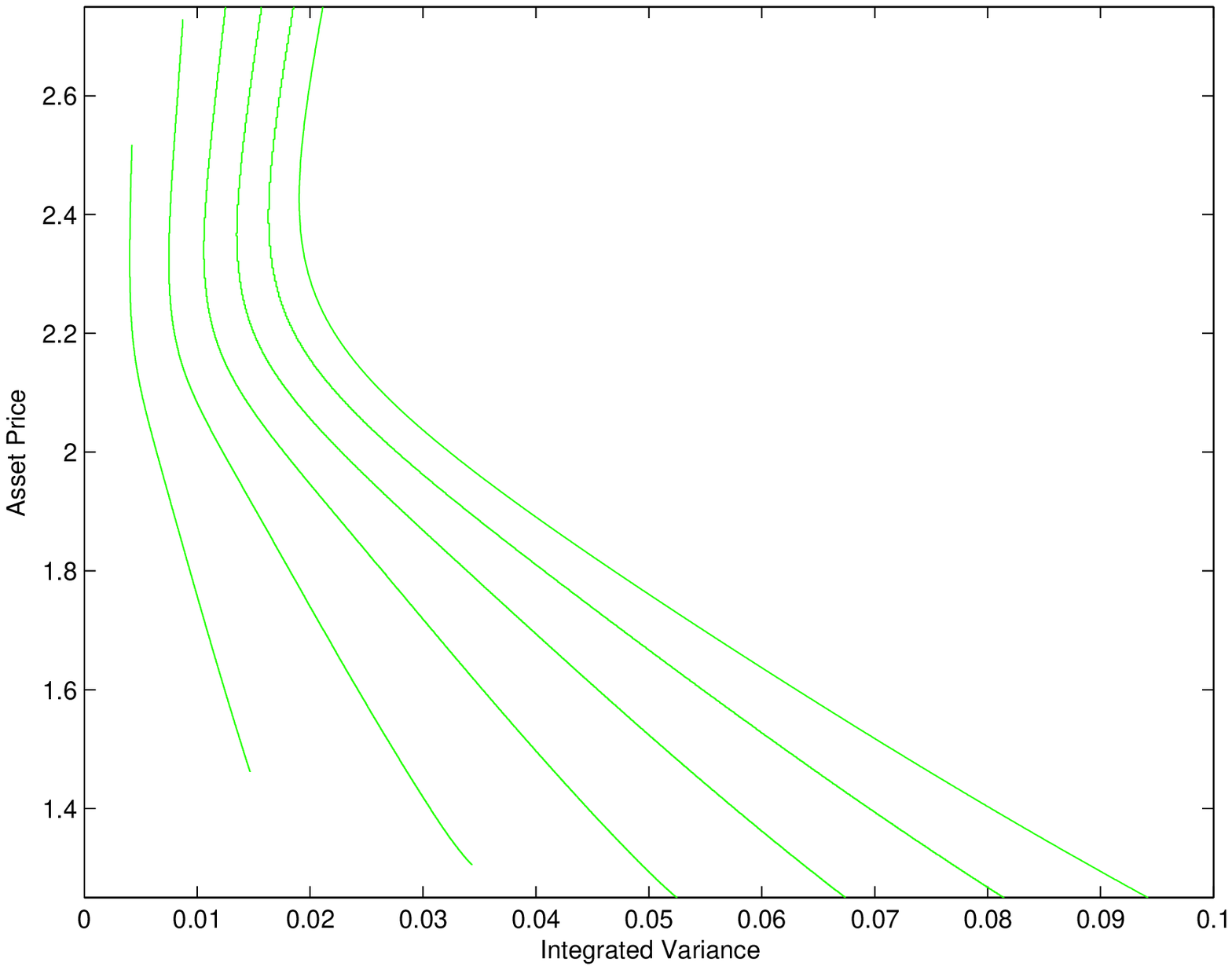}
  \end{tabular}
  \caption{We compare the barriers for multiple maturities. In this
    figure we compute the barriers at equal spaced maturities of the
    underlying Heston model (the last barrier corresponding to $T=1$)
    for  the Rost (left) and Root (right) cases.}
  \label{fig:BarriersInMaturity}
\end{figure}
 One can also consider the behaviour of the barriers in time.  In
Figure~\ref{fig:BarriersInMaturity} we plot the barriers for a
sequence of call prices with increasing maturity.

\begin{figure}[htb]
  \centering
  \includegraphics[width=.68\textwidth]{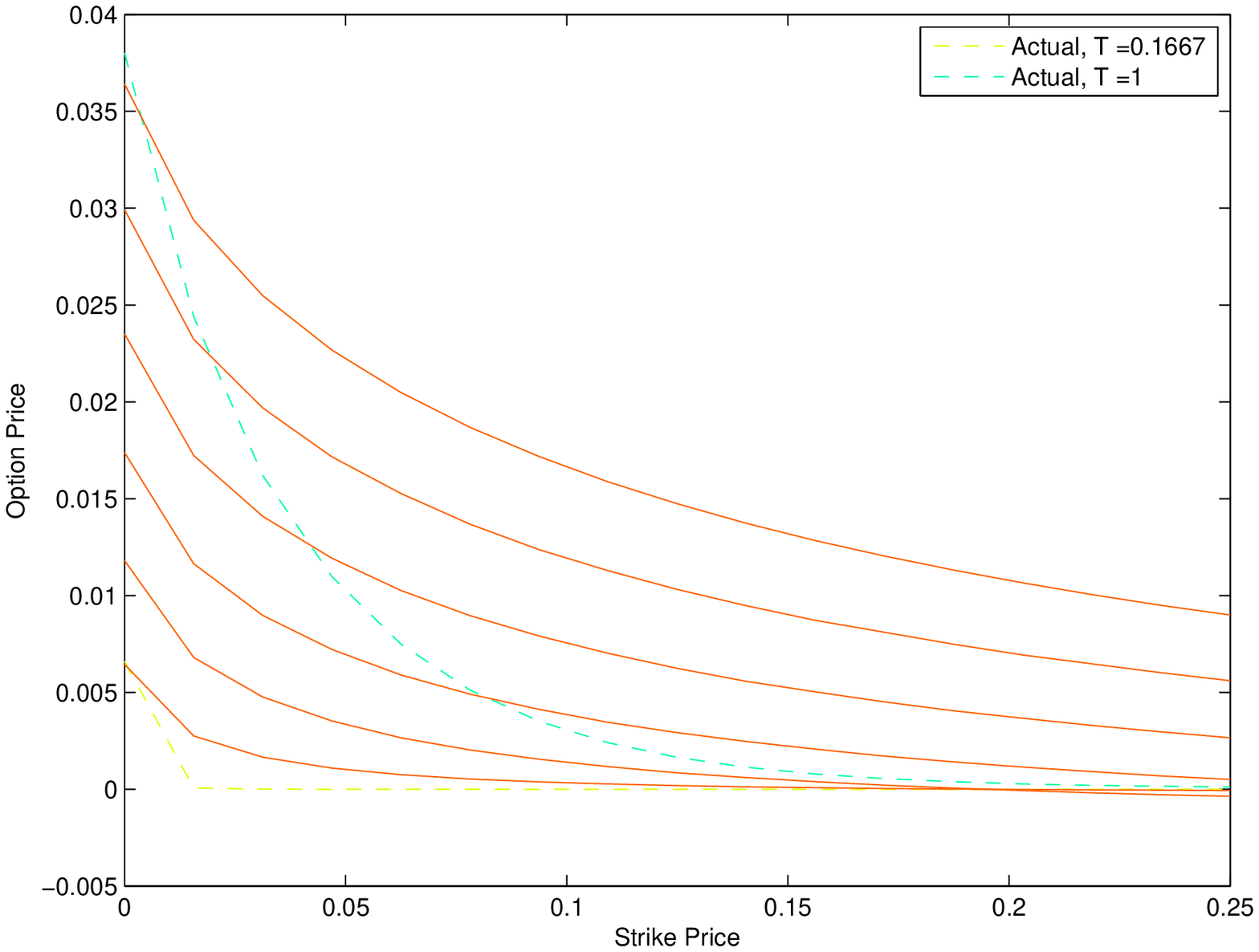}
  \caption{A plot of the upper bound on the price of a Variance call
    for different maturities (equally spaced, up to $T=1$). For
    comparison, we also plot the actual prices of the variance calls
    under the Heston model corresponding to the shortest and longest
    maturities.}
  \label{fig:VarCallInMaturity}
\end{figure}
 Of course, our interest lies in the implied bounds of options on
variance. We first consider the case of a variance call. In
Figure~\ref{fig:VarCallInMaturity} we display the upper bound on the
price of a variance call derived in Section~\ref{sec:optRost}. As
might be expected, there is a substantial difference between the upper
bound and the model-implied price.

\begin{figure}[htb]
  \centering
  \includegraphics[width=.68\textwidth]{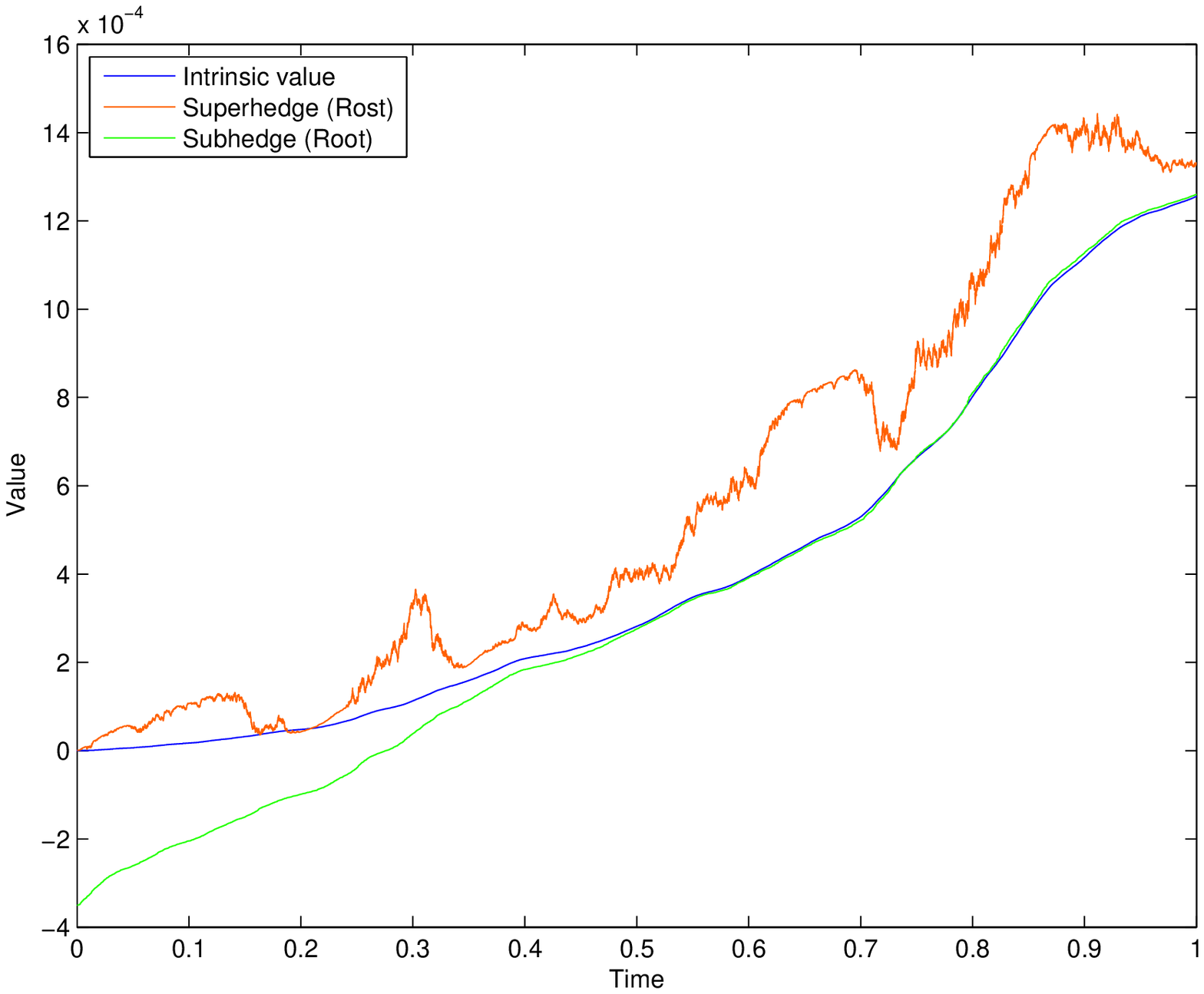}
  \caption{A plot of a realisation of the intrinsic value of the
    option ($F(\langle \ln S \rangle_t)$, and the value of the
    super-hedging and sub-hedging portfolios at time $t$, for $t \in
    [0,T]$. Here $F(t) = t(t-v_K)$, where $v_K = 0.01875$.}
  \label{fig:HedgeSimulate}
\end{figure}
To see how the hedges constructed in perform in a given realisation,
we can simulate a path, and compute the values of the super- and
sub-hedging strategies along the realisation. In this example, we
consider an option on variance with payoff $F(\langle \ln S
\rangle_T)$, where $F(t) = t (t \wedge v_K)$.  This is shown in
Figure~\ref{fig:HedgeSimulate}.
\begin{figure}[htb]
  \centering
  \includegraphics[width=.68\textwidth]{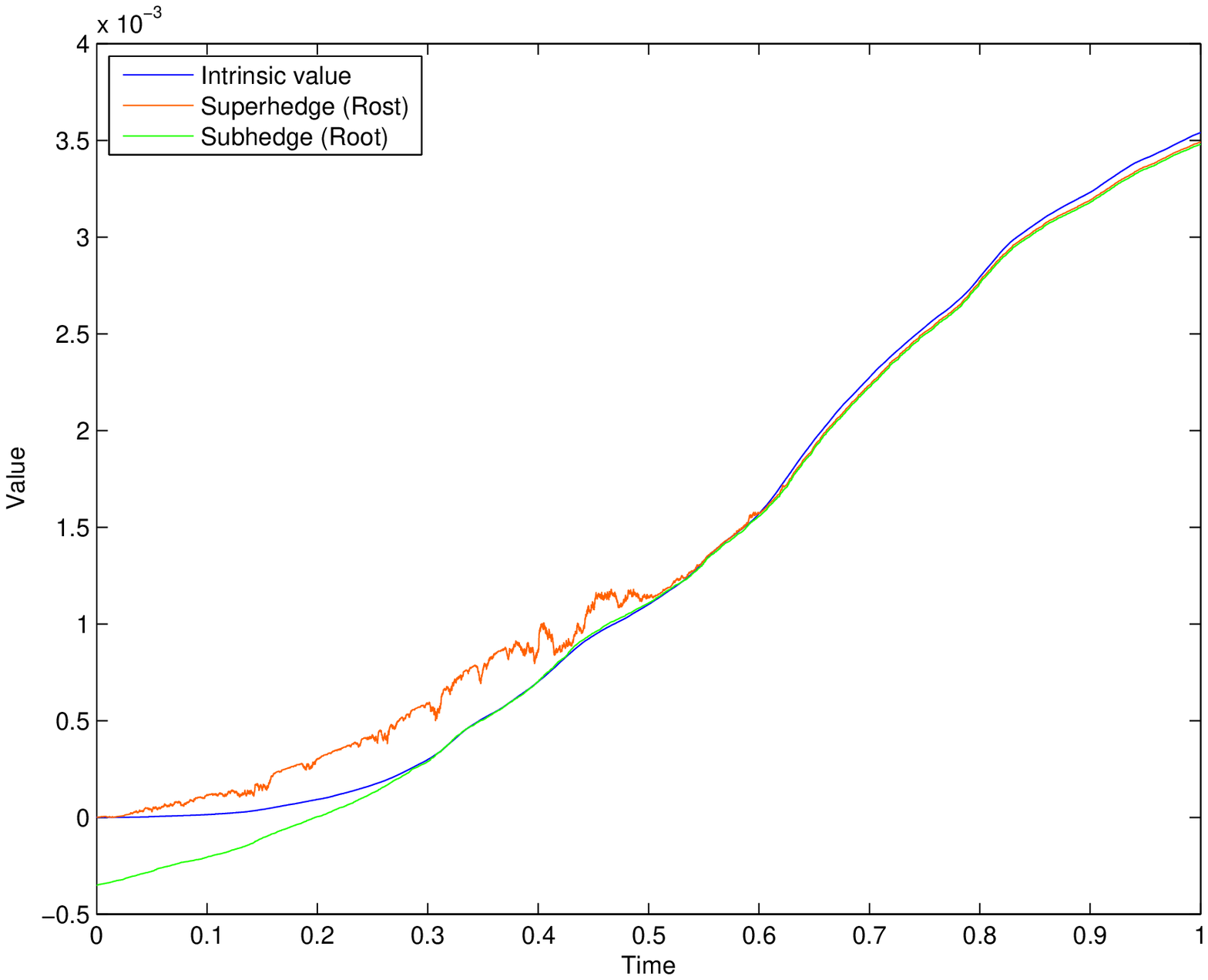}
  \caption{A plot of a realisation of the intrinsic value of the
    option, as in Figure~\ref{fig:HedgeSimulate}, and the
    corresponding sub- and super-hedges, where the realisation is
    taken from a different model to that under which the original
    hedge was constructed.}
  \label{fig:HedgeWrongModel}
\end{figure}
The main attraction of these hedging portfolios is that they remain
super/sub-hedges under a different model. For example, in
Figure~\ref{fig:HedgeWrongModel} we show how these hedges behave if
the path realisation comes from a Heston model with different
parameters. Here we set: $\rho' = 0.5, \theta' = 0.07$ and $\kappa' =
2.4$.
To conclude, we show that the sub- and super-hedges provide good
model-robustness by computing (empirically) the difference between the
payoff of an option on variance, and the corresponding super- or
sub-hedge. This is shown in Figure~\ref{fig:HedgeHistogram}, which
also shows the effect of model-misspecification on the distribution of
the hedging error.
\begin{figure}[htb]
  \centering
  \begin{tabular}{@{}cc@{}}
    \includegraphics[width=.48\textwidth]{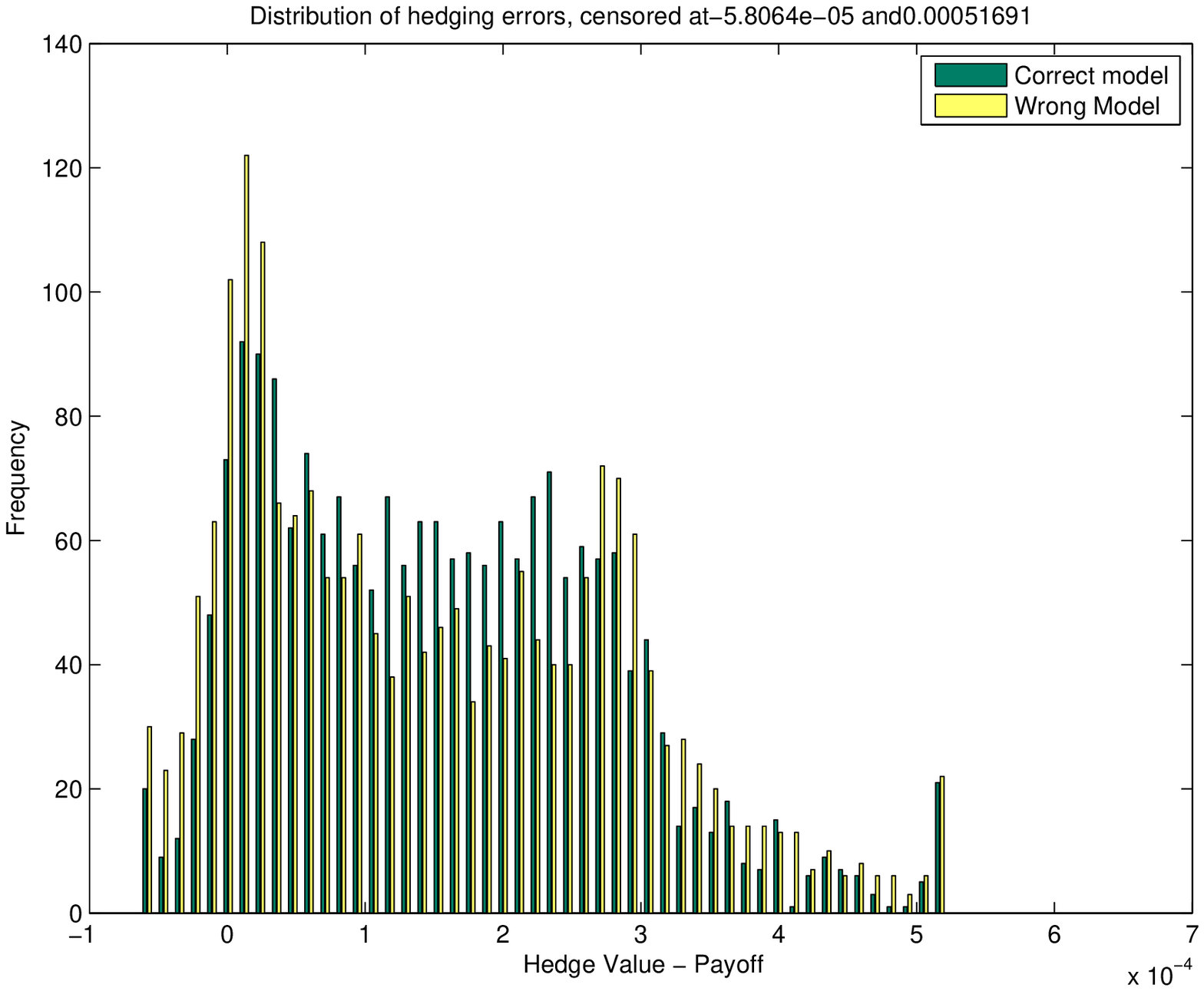}
    & 
    \includegraphics[width=.48\textwidth]{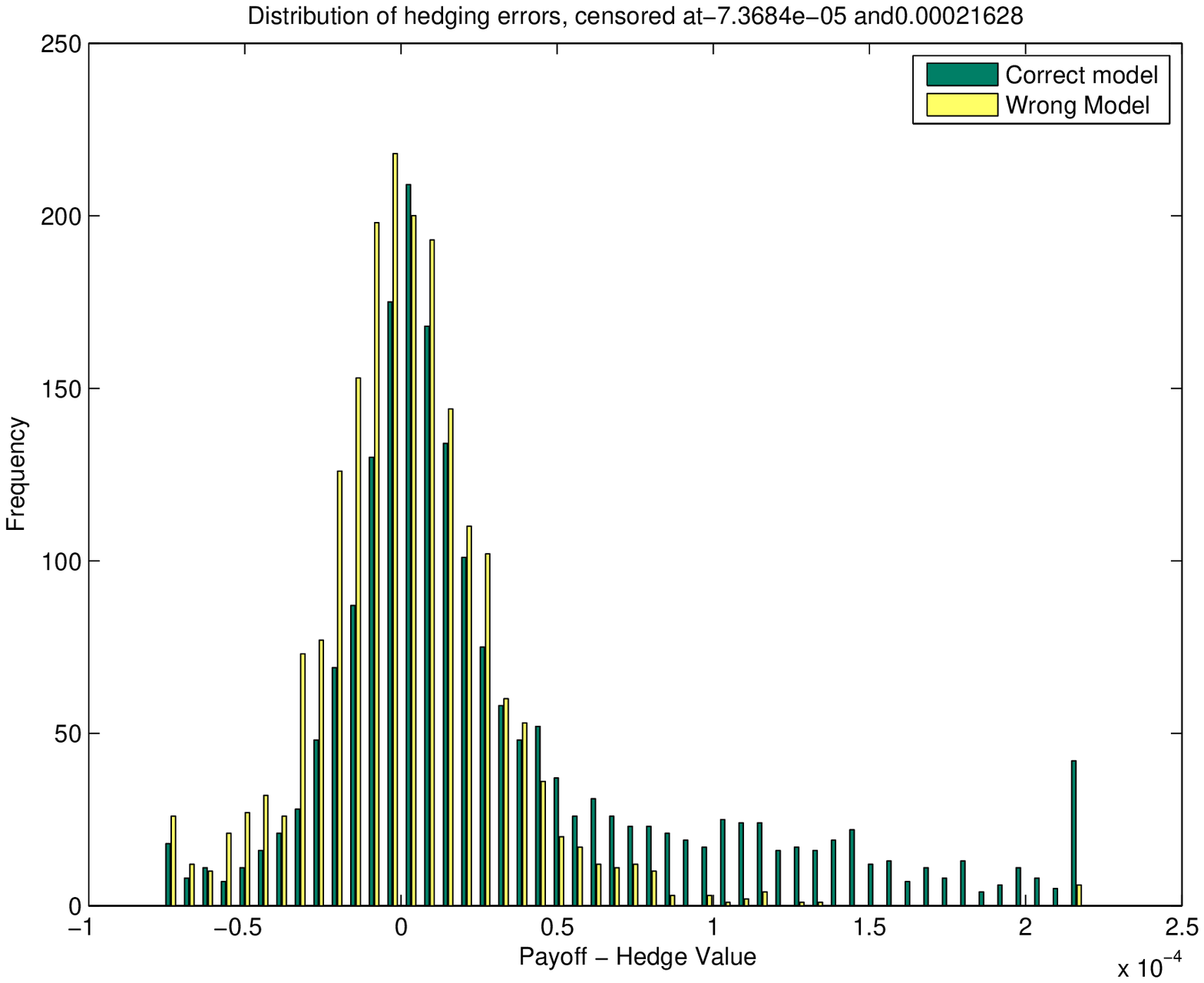}
  \end{tabular}
  \caption{The (empirically computed) distribution of the hedging
    error for the super-hedge (left) and sub-hedge (right). In the
    former case, a positive error represents a surplus in the hedge,
    while in the latter case, a positive error represents an
    underhedge. We also compare the tracking error with the tracking
    error for the same strategy when our realised path is determined
    by the mis-specified model of Figure~\ref{fig:HedgeWrongModel}.}
  \label{fig:HedgeHistogram}
\end{figure}

\section{Extremality and the Heston-Nandi model}\label{sec:extr-hest-nandi}

In this section, we consider a particular, commonly used model for
asset prices --- the Heston-Nandi model --- and show that it can have
particularly bad implications for the pricing of variance options.

The Heston-Nandi model \citep{HestonNandi:98} is the common Heston
stochastic volatility model \citep{Heston:93}, where the correlation
$\rho$ between the Brownian motions driving the asset and the
volatility processes is taken to be $-1$. Since in some asset classes,
$\rho \approx -1$ is not abnormal (\eg{} a Heston model calibrated to
the S\&P 500 typically has $\rho \approx -0.9 \pm 0.1$, for example,
see \citet{GuillaumeSchoutens:12}), and the pricing of options on
variance in Heston models is also common practice
\citep{JavaheriWilmottHaug:04}, this may substantially bias the prices
towards the extreme models.

The Heston model is given (under the risk-neutral measure) by:
\begin{equation} \label{eq:Heston}
  \begin{aligned}
    \df S_t & = r S_t \, \dt + \sqrt{v_t} S_t \, \df B_t,\\
    \df v_t & = \kappa (\theta - v_t) \, \dt + \xi \sqrt{v_t} \, \df
    \widetilde{B}_t,
  \end{aligned}
\end{equation}
where $B_t$ and $\widetilde{B}_t$ are Brownian motions with
correlation $\rho$. The Heston-Nandi model is the restricted case
where $\rho = -1$, and so $\widetilde{B}_t = - B_t$. Note that $v_t =
\sigma_t^2$ in our previous notation, so we are interested in options
on $\int_0^T v_t \, \dt$

The simplification $\rho = -1$ allows for the following observation:
using It\^o's Lemma, we know
\begin{align*}
  \df (\log(\me^{-rt}S_t)) & = - \half v_t \dif t + \sqrt{v_t}
  \dif B_t \\
  & = \left( \frac{\kappa \theta}{\xi} - \left( \frac{\kappa}{\xi} +
      \half\right)v_t \right) \dif t -\frac{1}{\xi} \df v_t.
\end{align*}
Solving, we see that
\begin{equation}\label{eq:HNProc}
  \log \left( \frac{\me^{-rT}S_T}{S_0} \right) = \frac{1}{\xi}(v_0 - v_T) +
  \frac{\kappa \theta}{\xi} T - \left( \frac{\kappa}{\xi}
    + \half\right) \int_0^T v_t \dif t.
\end{equation}
If we assume that the maturity time of our option, $T$, is
sufficiently large, since $v_T$ is mean reverting, $(v_T - v_0)\approx
(\theta - v_0)$ will be small in relation to the other terms on the
right-hand side.

If we temporarily ignore the $v_T-v_0$ term, \eqref{eq:HNProc} tells
us that, at time $T$, we have
\begin{equation*}
  \int_0^T  v_t \, \dt \approx \left( \frac{\kappa}{\xi}
    + \half\right)^{-1} \left[ \log\left( \frac{S_0}{S_T}\right) +
    \frac{\kappa \theta}{\xi} T \right].
\end{equation*}
Writing
\begin{equation} \label{eq:RTdefn} R_T(x) = \left( \frac{\kappa}{\xi}
    + \half\right)^{-1} \left[ \log\left( \frac{S_0}{x}\right) +
    \frac{\kappa \theta}{\xi} T \right],
\end{equation}
then we have
\begin{equation*}
  T \approx \inf\left\{ s \ge 0 : \int_0^s  v_t \, \dt \ge R_T(X_s)\right\}.
\end{equation*}
This describes a barrier stopping time, corresponding to a Root
stopping time, with
\begin{equation}
  \label{eq:HNDdefn}
  D_T = \left\{ (x,t) \in \R \times \R_+ : t < R_T(x)\right\}.
\end{equation}
So, ignoring the term in $v_T$, we might conjecture that the
corresponding model minimises the value of a derivative which is a
convex, increasing function of $v_T$ over all models with the same law
at time $T$.

This leads to the following result:
\begin{thm} \label{thm:HNThm} Let $M>0$ and suppose $\xi, \theta,
  \kappa, r>0, \xi \neq 2 \kappa$ are given parameters of a
  Heston-Nandi model, $\Q^{HN}$. Suppose $\mathcal{Q}_T$ is the class
  of models $\Q$ satisfying Assumption~\ref{ass:Price} and
  $\E^{\Q^{HN}}(S_T-K)_+ = \E^\Q(S_T -K)_+$ for all $K\ge 0$.

  Then there exists a constant $\kappa$, depending only on $M$ and the
  parameters of the Heston-Nandi model, such that for all convex,
  increasing functions $F(t)$ with suitably smooth derivative $f(t) =
  F'(t)$ such that $f(t), f'(t) \le M^*$, and for all $T \ge 0$
  \begin{equation}
    \label{eq:HNThm}
    \E^{\Q^{HN}} F\left(\langle \log S\rangle_T\right) \le \inf_{\Q \in
      \mathcal{Q}_T} \E^\Q F\left(\langle \log S\rangle_T\right) + \kappa.
  \end{equation}
\end{thm}

Note that the strength of the result depends on the fact that the
constant $\kappa$ is independent of both $T$ and $F$. In particular,
$\langle \log S\rangle_T$ should be both growing in $T$ and increasing
in variance as $T$ increases. That this does not appear in the bound
leads us to claim that $\Q^{HN}$ is {\it asymptotically optimal}.

In fact, the continuity assumptions can be trivially relaxed, and this
leads to the simple corollary:
\begin{cor}
  The conclusions of Theorem~\ref{thm:HNThm} hold where the class of
  functions $F$ considered is the set of variance call payoffs:
  $F_K(t) = (t-K)_+$ for all $K \in \R^+$ and all maturity dates
  $T>0$.
\end{cor}
\begin{proof}
  For fixed $M^*$, the function $F_K(t)$ can be approximated uniformly
  from above and below by a suitably smooth function satisfying the
  conditions of Theorem~\ref{thm:HNThm}, independent of $K$. The
  result follows.
\end{proof}

The above result demonstrates that the seemingly strong assumptions on
the function $F(t)$ required in Theorem~\ref{thm:HNThm} are not a big
restriction: by allowing a slightly larger constant, we can consider
the class of functions which can be approximated by such functions
uniformly. As we will see, the exact smoothness requirements on $f(t)$
are that $f(t)$ has a H\"older continuous second derivative (although
we believe that this assumption could be relaxed).

Our arguments rely on the construction of a barrier, and the proof of
optimality described in \cite{CoxWang:11}. We recall some important
definitions here. We suppose that we are given a barrier function
$R_T(x)$ as defined in \eqref{eq:RTdefn}, and consider the geometric
Brownian motion $(\wti{X}_t)$ on this domain, with corresponding
hitting time $\tau_{D_T}$, where $D_T$ is as defined in
\eqref{eq:HNDdefn}. Then we define the function
\begin{equation} \label{eq:M2defn} M(x,t) = \E^{(x,t)}f(\tau_{D_T})
\end{equation}
and observe that (under the assumptions of Theorem~\ref{thm:HNThm}),
we have $M(x,t)$ bounded. Since we consider the case where $\wti{X}$
is geometric Brownian motion, we can assume $\sigma(x) = x$ in the
formulae from \cite{CoxWang:11}.

Now define a function $Z(x)$ by:
\begin{equation} \label{eq:Zdefn} Z(x) = 2 \int_{S_0}^x \int_{S_0}^y
  \frac{M(z,0)}{z^2} \dif z \dif y.
\end{equation}
So in particular, we have $Z''(x) = 2\frac{M(x,0)}{x^2}$ and $Z(x)$ is
a convex function. Define also:
\begin{equation} \label{eq:Gdefn} G(x,t) = \int_0^t M(x,s) \dif s -
  Z(x),
\end{equation}
and
\begin{equation}
  \label{eq:Hdefn}
  H(x) = \int_0^{R_T(x)} (f(s) - M(x,s)) \, \ds + Z(x).
\end{equation}
Then (\citet[Proposition~5.1]{CoxWang:11}) for all $(x,t) \in \R_+
\times \R_+$:
\begin{equation}
  \label{eq:GFineq}
  G(x,t) + H(x) \le F(t)
\end{equation}
with equality when $t = R_T(x)$.  In addition, if for any $T > 0$:
\begin{equation}
  \label{eq:qvbddcond}
  \E \left[ \int_0^T Z'(\wti{X}_s)^2 \sigma(\wti{X}_s)^2 \dif s \right] <
  \infty, \quad \E Z(\wti{X}_0) < \infty,
\end{equation}
then the process
\begin{equation} \label{eq:Gmart} G(\wti{X}_{t\wedge \tau_{D_T}},t
  \wedge \tau_{D_T}) \text{ is a martingale,}
\end{equation}
and
\begin{equation}\label{eq:Gsubmart}
  G(\wti{X}_{t},t) \text{ is a submartingale.}
\end{equation}

We collect some useful properties of these functions in the following
lemma:
\begin{lem} \label{lem:PropofFn} Under the assumptions of
  Theorem~\ref{thm:HNThm}, the functions $Z(x), H(x)$ and $G(x,t)$ as
  defined above have the following properties:
  \begin{enumerate}
  \item\label{item:2} $|\pd{G}{t}(x,t)| \le M^*$ for all $(x,t) \in
    \R_+\times \R_+$.

  \item \label{item:3} $G(\wti{X}_t,t)$ is a submartingale, with
    decomposition:
    \begin{align*}
      G(\wti{X}_t,t) & = G(\wti{X}_0,0) + \int_0^t \wti{X}_s
      \left(\int_0^{R_T(\wti{X}_s)\wedge t}\pd{M}{x}(\wti{X}_s,r) \,
        \dif r
        -Z'(\wti{X}_s)\right) \, \dif \wti{X}_s \\
      & \quad \quad {} - \int_0^t \gamma(\wti{X}_s)
      \indic{s>R_T(\wti{X}_s)} \, \ds,
    \end{align*}
    where
    \begin{equation*}
      \gamma(x) = f(R_T(x))-\half x^2 \pd{M}{x}(x,R_T(x)-)
      R_T'(x) \ge 0
    \end{equation*}
    is a bounded function.
  \end{enumerate}
    
\end{lem}

\begin{proof}[Proof of Theorem~\ref{thm:HNThm}]
  Let $v_t$ be the squared volatility process for the Heston-Nandi
  price process $S_t$, and suppose we fix $T>0$ (although we will want
  our constants to be independent of $T$). Define the time-change
  process $\tau_T = \int_0^T v_t \, \dt$, and let $A_t$ be the
  right-inverse of $\tau_t$. In particular, if we define as usual
  $\wti{X_t} = \me^{-rA_t} S_{A_t}$, then $\wti{X}_t$ is a geometric
  Brownian motion with fixed law $\mu_T$ at time $\tau_T$. Using
  \eqref{eq:HNProc} in \eqref{eq:RTdefn} we get:
  \begin{equation}\label{eq:RTformula}
    R_T(\wti{X}_{\tau_T}) = \int_0^T v_s \, \ds +\frac{1}{\xi} \left( \frac{\kappa}{\xi}
      + \half\right)^{-1} (v_T-v_0).
  \end{equation}
  Since the variance process $v_s$ is mean reverting $\E|v_T-v_0|$,
  can be bounded uniformly for all $T$ by some constant depending only
  on the parameters of the model and so in particular, there exists a
  constant $\kappa_1$ such that
  \begin{equation*}
    \E |R_T(\wti{X}_{\tau_T}) - \tau_T| < \kappa_1.
  \end{equation*}
  From the bound on $f(t)$, it then follows that:
  \begin{equation} \label{eq:Fbound}
    \E[|F(\tau_T)-F(R_T(\wti{X}_{\tau_T}))|] \le M^* \kappa_1.
  \end{equation}
  Similarly, using Lemma~\ref{lem:PropofFn}.\ref{item:2}, we get
  \begin{equation*}
    \E[|G(\wti{X}_{\tau_T},\tau_T)-G(\wti{X}_{\tau_T}, R_T(\wti{X}_{\tau_T}))|] \le M^* \kappa_1.
  \end{equation*}
  
  In addition, using the decomposition from
  Lemma~\ref{lem:PropofFn}.\ref{item:3}, and noting that $\gamma(x)$
  is bounded above by a constant, $\kappa_2$ say, we have:
  \begin{equation*}
    \E G(\wti{X}_{\tau_T},\tau_T) \le  G(\wti{X}_0,0) +
    \kappa_2\E\left[\int_0^{\tau_T} \indic{s > R_T(\wti{X}_s)} \, \ds\right].
  \end{equation*}
  Observe from the definition of $\tau_t$, \eqref{eq:RTdefn} and
  \eqref{eq:HNProc} evaluated at a general time $t = A_s$:
  \begin{align*}
    \{ R_T(\wti{X}_s) < s \} & = \left\{ \left( \frac{\kappa}{\xi} +
        \half\right)^{-1} \left[ \frac{\kappa\theta}{\xi} (T-A_s) +
        \frac{1}{\xi} \left(v_{A_s}-v_0\right)\right] \le 0 \right\} \\
    & = \left\{A_s \ge \frac{\kappa\theta T + v_{A_s} -
        v_0}{\kappa\theta}\right\} \subseteq \left\{A_s \ge
      \frac{\kappa\theta T - v_0}{\kappa\theta}\right\}.
  \end{align*}
  Hence
  \begin{align*}
    \E\left[\int_0^{\tau_T} \indic{s > R_T(\wti{X}_s)} \, \ds\right] &
    \le \E \left[ \int_0^{\tau_T} \indic{A_s \ge
        T-\frac{v_0}{\kappa\theta}} \, \ds \right] = \E \left[
      \int_{\tau_{T-\frac{v_0}{\kappa\theta}}}^{\tau_T} \, \ds
    \right] \\
    & = \E \left[ \tau_T - \tau_{T-\frac{v_0}{\kappa\theta}} \right] =
    \E \left[ \int_{T-\frac{v_0}{\kappa\theta}}^{T} v_s \, \ds
    \right].
  \end{align*}
  Again, since $v_s$ is mean reverting, the right-hand-side can be
  bounded independently of $T$, and so
  \begin{equation} \label{eq:Gbound} \Ep{G(\wti{X}_{\tau_T},\tau_T)}
    \le G(\wti{X}_0,0) + \kappa_2
  \end{equation}
  for some constant $\kappa_2$.

  Now, using \eqref{eq:Fbound} and the fact that \eqref{eq:GFineq}
  holds with equality when $t=R_T(x)$, we have:
  \begin{align*}
    \E \left[ F(\tau_T) \right] & \le \E\left[
      F(R_T(\wti{X}_{\tau_T}))\right] +
    M^*\kappa_1\\
    & = \E\left[G(\wti{X}_{\tau_T},R_T(\wti{X}_{\tau_T}))\right] +
    \Ep{H(\wti{X}_{\tau_T})} +
    M^*\kappa_1\\
    & = \E\left[G(\wti{X}_{\tau_T},\tau_T)\right] +
    \Ep{H(\wti{X}_{\tau_T})} +
    2M^*\kappa_1\\
    & \le G(\wti{X}_0,0) + \Ep{H(\wti{X}_{\tau_T})} +
    2M^*\kappa_1+\kappa_2
  \end{align*}
  
  It remains for us to show that $\Ep{F(\wti{X}_\sigma)} \ge
  G(\wti{X}_0,0) + \Ep{H(\wti{X}_{\tau_T})} = G(\wti{X}_0,0) +
  \Ep{H(\wti{X}_{\sigma})}$ for any stopping time $\sigma$ with
  $\wti{X}_\sigma \sim \wti{X}_{\tau_T}$.

  We consider a localising sequence, $\sigma_N \uparrow \sigma$, and
  note that we then have: $\Ep{G(\wti{X}_{\sigma_N},\sigma_N)} \ge
  G(\wti{X}_0,0)$, since $G$ is a submartingale, and in addition,
  $\Ep{F(\sigma_N)} \uparrow \Ep{F(\sigma)}$ since $F(\cdot)$ is
  increasing. On account of \eqref{eq:GFineq}, it remains only to show
  $\Ep{H(\wti{X}_{\sigma})} \le \Ep{H(\wti{X}_{\sigma_N})}$. We first
  observe that $f$ is an increasing and bounded function, and if $f(t)
  = f(\infty)$ for all $t \ge t_0$, for some $t_0\in \R_+$, $M(x,t) =
  f(t)$ for all $t \ge t_0$. Since also $Z(x) \ge 0$ ($Z(x)$ is convex
  with $Z(S_0) = Z'(S_0) = 0$), we must have $H(x)$ bounded below. We
  can therefore apply Fatou's Lemma to deduce
  $\Ep{H(\wti{X}_{\sigma})} \le \Ep{H(\wti{X}_{\sigma_N})}$. To remove
  the assumption on $f(t)$, we observe that, by smoothly truncating
  $f$, we can approximate $F$ from below by an increasing sequence
  $F_N$ of functions which each have constant derivative, and such
  that $\Ep{F_N(\wti{X}_{\tau_T})} \uparrow
  \Ep{F(\wti{X}_{\tau_T})}$. Since each approximation satisfies the
  bound, the same must be true in the limit.
  

\end{proof}

\begin{proof}[Proof of Lemma~\ref{lem:PropofFn}]
  We first show \ref{item:2}. Observe that $M(x,R_T(x)) = f(R_T(x))$,
  so in particular, $M$ is continuous, and $f(t) \le M(x,t) \le M^*$
  since $M(x,t)$ is increasing in $t$. It follows immediately that
  $\pd{G}{t}(x,t) = M(x,t)$ is continuous and bounded, and in fact, is
  non-negative.

  For \ref{item:3} we aim to use \citet[Theorem~3.1]{Peskir:05}. We
  note that $M(x,t)$ is $C^{2,1}$ in ${D_T}$ since it is a martingale
  (in particular, $M$ is the (unique bounded) solution to a parabolic
  initial-value boundary problem). In fact, by
  \citet[Theorem~5.14]{Lieberman:96}, if we assume that $f''(t)$ is
  bounded in a H\"older norm, it follows that $M(x,t)$ has
  H\"older-bounded first and second spatial derivatives, and first
  time derivative. It is easy to check that $G(x,t) = \int_0^t M(x,s)
  \, \ds - Z(x)$ is also $C^{2,1}$ in ${D_T}$ as a
  consequence. Moreover, computing explicitly, we see that
  $\pd{G}{x}(x,t) = \int_0^t \pd{M}{x}(x,s) \, \ds - Z'(x)$ is also
  continuous on $\bar{D}_T$. If we write $C = \{(x,t) : R_T(x) <
  t\}$, so $\bar{C} \cup \bar{D}_T = \R_+ \times \R_+$ with boundary
  $R_T(x) = t$, we see that on $C$,
  \begin{equation*}
    G(x,t) = \int_0^{R_T(x)} M(x,s) \, \ds + \int_{R_T(x)}^t f(s) \, \ds - Z(x)
  \end{equation*}
  and again, $G$ is $C^{2,1}$ in $C$, and $\pd{G}{x}(x,t)$ is
  continuous on $\bar{C}$.

  Considering the conditions required for Theorem~3.1 of
  \cite{Peskir:05}, we observe that (3.18), (3.19), (3.26), (3.30) and
  (3.33) of \cite{Peskir:05} have now been shown, and so the theorem
  holds. Moreover, since the first spatial derivative of $G$ is
  continuous across the boundary, we do not get a local-time term on
  the boundary. Computing $\pd{G}{t}(x,t) + \half \sigma(x)^2
  \pd{^2G}{x^2}(x,t)$ results in the expression for $\gamma(x)$
  stated, and we observe that the boundary function $R_T(x)$ is a
  decreasing function of $x$, which implies in turn that
  $\pd{M}{x}(x,t)$ is positive at the boundary (since $f$ is
  increasing), so $\gamma(x)\ge0$.

  We finally show that $\gamma(x)$ is bounded. By assumption, $f$ is
  bounded, so we need only consider the second term. We have
  $M(x,R_T(x)) = f(R_T(x))$, and differentiating (recall that the
  derivatives of $M(x,t)$ on ${D_T}$ are H\"older continuous, and so
  extend continuously to the boundary) and rearranging we get:
  \begin{equation*}
    \pd{M}{x}(x,R_T(x)) =  R_T'(x)
    \left(
      f'(R_T(x))- \pd{M}{t}(x,R_T(x))
    \right)
  \end{equation*}
  Observing that (via a standard coupling argument)
  \[\Pr^{(x,t)}(\tau_{D_T} > t) \ge \Pr^{(x,t+\delta t)}(\tau_{D_T} >
  t+\delta t)\] whenever $(x,t) \in {D_T}$ (the later path sees a
  `bigger' stopping region). It follows that $\E^{(x,t+\delta t)}
  \left[ f(\tau_{D_T}) \right] \le \E^{(x,t)} \left[
    f(\tau_{D_T}+\delta t) \right]$, and therefore that
  $\pd{M}{t}(x,t) \le \sup_{t \in \R_+} f'(t) \le M^*$.

  Recalling finally that $\sigma(x) = x$, and observing that $R_T'(x)
  = - ( \kappa/\xi + 1/2)^{-1} x^{-1}$, we conclude that
  $\gamma(x)$ is indeed bounded.
\end{proof}

\begin{figure}[htb]
  \centering
  \includegraphics[width=.68\textwidth]{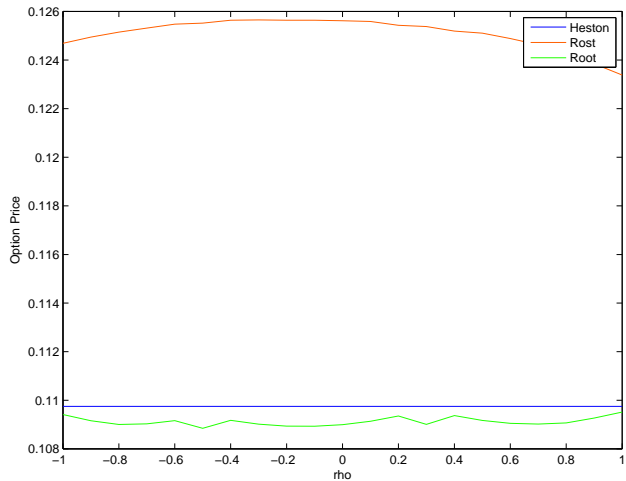}
  \caption{The upper and lower model-independent bounds on the price
    of a variance call, plotted as functions of the correlation
    between the asset and volatility processes. The constant line
    represents the price of the variance call under the Heston model
    --- this is constant, since the price is unaffected by the choice
    of $\rho$.}
  \label{fig:RhoBehaviour}
\end{figure}
We finish with some numerical evidence to support our
conjecture. Figure~\ref{fig:RhoBehaviour} shows the upper and lower bounds on
the price of a variance call, as determined using the numerical methods of
Section~\ref{sec:numerical-results}, seen as a function of the parameter
$\rho$. In this example, we use the same parameters as before, but with
$T=4$. It is notable that the lower bound and the price arising from the Heston
model are certainly close. It is also interesting to observe that it is not
only the Heston-Nandi model that seems to be close to extremal; rather this
seems to be a more general property of the Heston model. A good explanation of
this fact eludes us, but better understanding of this behaviour would appear to
be both practically relevant, and theoretically interesting.




\renewbibmacro*{in:}{} 

\printbibliography

\end{document}
